\documentclass[a4paper,USenglish,numberwithinsect]{lipics-v2016}
\usepackage{textcomp}

\usepackage{etoolbox}

\providebool{techreport}
\setbool{techreport}{true}

\usepackage{array}
\usepackage{graphicx}
\usepackage[justification=Centering]{subfig}
\usepackage{algorithm}
\usepackage{algorithmic}
\usepackage{tikz}
\usepackage{xcolor}
\usepackage{multirow}
\usepackage{paralist}
\usepackage{url}
\usepackage{etoolbox}
\usepackage{cancel}

\usepackage{mathtools}
\usepackage{marvosym}

\usetikzlibrary{calc,shadows,patterns,shapes,arrows,decorations.pathmorphing,backgrounds,positioning,fit,plotmarks}

 \def\myendproof{{\ \vbox{\hrule\hbox{%
   \vrule height1.3ex\hskip0.8ex\vrule}\hrule }}\par}
 
\newcommand{\floor}[1]{\left\lfloor{#1}\right\rfloor}
\newcommand{\ceiling}[1]{\left\lceil{#1}\right\rceil}
\newcommand{\setof}[1]{\left\{{#1}\right\}}

\newcommand{\dbf}[0]{\mbox{\sc dbf}}
\newcommand{\load}[0]{\mbox{\sc load}}

\usepackage{array}
\newcolumntype{L}[1]{>{\raggedright\let\newline\\\arraybackslash\hspace{0pt}}m{#1}}
\newcolumntype{C}[1]{>{\centering\let\newline\\\arraybackslash\hspace{0pt}}m{#1}}
\newcolumntype{R}[1]{>{\raggedleft\let\newline\\\arraybackslash\hspace{0pt}}m{#1}}

\tikzset{
    task/.style={shade, shading=radial, rectangle,minimum height=.1cm,
        inner color=#1!20, outer color=#1!60!gray},
    task1/.style={task=yellow, minimum width=13mm},
    task2/.style={task=orange, minimum width=13mm},
    task3/.style={task=red, minimum width=13mm},
    task4/.style={task=green, minimum width=13mm},
    task5/.style={task=blue, minimum width=13mm},
    task6/.style={task=purple, minimum width=13mm},
    task7/.style={task=cyan, minimum width=13mm},
    task8/.style={task=pink, minimum width=13mm},
}

\floatname{algorithm}{Algorithm}

\EventEditors{Sebastian Altmeyer}
\EventNoEds{1}
\EventLongTitle{30th Euromicro Conference on Real-Time Systems (ECRTS 2018)}
\EventShortTitle{ECRTS 2018}
\EventAcronym{ECRTS}
\EventYear{2018}
\EventDate{July 3--6, 2018}
\EventLocation{Barcelona, Spain}
\EventLogo{}
\SeriesVolume{XX}
\ArticleNo{YY}

\begin{document}

\title{Push Forward: Global Fixed-Priority Scheduling of
Arbitrary-Deadline Sporadic Task Systems}
\titlerunning{Push Forward: Arbitrary-Deadline Sporadic Task Systems}

\author[1]{Jian-Jia Chen}
\author[2]{Georg von der Br\"uggen}
\author[3]{Niklas Ueter}

\affil[1]{TU Dortmund University, Germany\\
  \texttt{jian-jian.chen@tu-dortmund.de}}
\affil[2]{TU Dortmund University, Germany\\
\texttt{georg.von-der-brueggen@tu-dortmund.de}}

\affil[3]{TU Dortmund University, Germany\\
  \texttt{niklas.ueter@tu-dortmund.de}}

\authorrunning{J.-J. Chen, G. von der Br\"uggen, and N. Ueter} 

\Copyright{Jian-Jia Chen, Georg von der Br\"uggen,  and Niklas Ueter}





\maketitle
\begin{abstract}
  The sporadic task model is often used to analyze recurrent execution
  of identical tasks in real-time systems.  A sporadic task defines an
  infinite sequence of task instances, also called jobs, that arrive
  under the minimum inter-arrival time constraint.  To ensure the
  system safety, timeliness has to be guaranteed in addition to
  functional correctness, i.e., all jobs of all tasks have to be
  finished before the job deadlines.  We focus on analyzing
  arbitrary-deadline task sets on a homogeneous (identical) multiprocessor system
  under any given global fixed-priority scheduling approach and 
  provide a series of schedulability tests with different tradeoffs
  between their time complexity and their accuracy. Under the
  arbitrary-deadline setting, the relative deadline of a task can be
  longer than the minimum inter-arrival time of the jobs of the task.
  We show that global deadline-monotonic (DM) scheduling
  has a speedup bound of $3-1/M$ against any optimal
  scheduling algorithms, where $M$ is the number of identical processors, 
  and prove that this bound is
  asymptotically tight.
\end{abstract}

\section{Introduction}
\label{sec:intro}

The sporadic task
model  is
the basic task model in real-time systems, where  
each task $\tau_i$ releases an infinite number of \emph{task
instances} (\emph{jobs}) under its \emph{minimum inter-arrival time} 
(\emph{period}) $T_i$ and is further characterized by its \emph{relative
deadline} $D_i$ and its \emph{worst-case execution time} $C_i$.
The sporadic task model has been widely adopted 
in real-time systems.  
A sporadic task defines an infinite sequence of task instances, also
called 
\emph{jobs}, 
that arrive under the minimum inter-arrival time
constraint, i.e.,   
 any two consecutive releases of jobs of task $\tau_i$ 
 are temporally
 separated by at least $T_i$.  When a job of task $\tau_i$ arrives at time $t$,
it must finish no later than its \emph{absolute deadline} $t+D_i$.
 If 
 all tasks release their jobs strictly periodically with period $T_i$, 
 the task model
is the well-known Liu and Layland task model \cite{liu73scheduling}.
A sporadic task set is
called with 1)~\emph{implicit deadlines}, if the relative deadlines 
are equal to their minimum inter-arrival times,
2)~\emph{constrained deadlines}, if the minimum inter-arrival times are
no less than their relative deadlines, and 3)~\emph{arbitrary
  deadlines}, otherwise.

To schedule such task sets on a multiprocessor platform,
three paradigms have been widely adopted: 
partitioned, global, and
semi-partitioned multiprocessor scheduling.  The \emph{partitioned} scheduling
approach partitions the tasks statically among the available
processors, i.e., a task executes all its jobs on the assigned processor. 
 The \emph{global} scheduling approach allows a job to migrate from one
 processor to another at any time. The \emph{semi-partitioned} scheduling
 approach decides whether a task is divided into subtasks statically and how each task/subtask is then
assigned to a processor. 
A comprehensive survey of
multiprocessor scheduling for real-time systems can be found in
\cite{DBLP:journals/csur/DavisB11}. 

We focus on \emph{global fixed-priority preemptive
  scheduling} on $M$ identical processors, i.e., unique fixed priority levels
  are statically assigned to the tasks and at any point in
time the $M$ highest-priority jobs in the ready queue are
executed. Hence, the schedule is \emph{workload-conserving}. 
The response time of a job is defined as
its finish time minus its arrival time. The worst-case response time
of a task is an upper bound on the response times of all the jobs of
the task and can be derived by a \emph{(worst-case) response time analysis}
for a sporadic task under a  given scheduling algorithm.
Verifying whether a set of sporadic tasks can meet
their deadlines by a scheduling algorithm is called a
\emph{schedulability test}, i.e., verifying if the \emph{(worst-case) response
time} is smaller than or equal to the \emph{relative
deadline}. 


\subsection{Related Work}
\label{sec:related-work}

For uniprocessor
systems, i.e, M=1, 
the exact schedulability test and the (tight) worst-case response time
analysis by using \emph{busy intervals} were provided by
Lehoczky~\cite{DBLP:conf/rtss/Lehoczky90}.
Several approaches have been
proposed to reduce the time complexity, e.g., \cite{sjodin1998improved}. 
Bini and Buttazzo \cite{DBLP:journals/tc/BiniB04} proposed a framework of schedulability
tests that can be tuned to balance the time complexity and the
acceptance ratio of the schedulability test for uniprocessor sporadic
task systems. To achieve polynomial-time schedulability tests and response
time analyses, Lehoczky~\cite{DBLP:conf/rtss/Lehoczky90} 
proposed a utilization upper bound for a set of sporadic arbitrary-deadline tasks under fixed-priority
scheduling.  
  The linear-time response-time bound for fixed-priority systems was first 
  proposed
by Davis and Burns~\cite{DavisRTSS2008}, and later 
improved by Bini et
al.~\cite{DBLP:journals/tc/BiniNRB09,bini-RTSS2015} and Chen et
al. \cite{DBLP:journals/corr/abs-k2q}.
The computational complexity of the schedulability test problem and the worst-case response time analysis in uniprocessor 
systems for different variances can be found in \cite{DBLP:conf/soda/BonifaciCMM10,DBLP:conf/soda/EisenbrandR10,EisenbrandR08,DBLP:conf/ecrts/Ekberg015,DBLP:conf/rtss/Ekberg015}.

In this paper, we will implicitly assume multiprocessor systems, i.e., $M \geq
2$. Many results are known for constrained-deadline ($D_i\leq T_i$) and
implicit-deadline task systems \mbox{($D_i=T_i$)} on identical
multiprocessor platforms, 
 e.g.,
\cite{DBLP:conf/rtss/AnderssonBJ01,baruah2007techniques,DBLP:conf/rtss/GuanSYY09,DBLP:conf/opodis/Andersson08,DBLP:journals/rts/BaruahBMS10,DBLP:journals/corr/abs-k2q}. 
For details, please refer to the survey by Davis and Burns
\cite{DBLP:journals/csur/DavisB11}. Unfortunately, deriving exact
schedulability tests under multiprocessor global scheduling is
much harder than deriving them for uniprocessor systems due to the lack of concrete
worst-case scenarios that can be constructed efficiently.  
Most results in the literature focus on sufficient schedulability tests. 
Exceptions
are the exhaustive search under discrete time parameters by Baker and
Cirinei \cite{DBLP:conf/opodis/BakerC07}, finite automata under
discrete time parameters by Geeraerts et
al. \cite{geeraerts2013multiprocessor}, and hybrid finite automata by
Sun and Lipari \cite{DBLP:journals/rts/SunL16}. Specifically,
Geeraerts et al. \cite{geeraerts2013multiprocessor} showed that the
schedulability test formulation by Baker and Cirinei \cite{DBLP:conf/opodis/BakerC07} 
is {\sc Pspace}-Complete.

Regarding global fixed-priority scheduling for arbitrary-deadline task
systems, several sufficient schedulability tests and safe worst-case
response time analyses have been proposed, 
e.g.,~\cite{baker2006analysis,DBLP:conf/opodis/BakerC07,DBLP:conf/opodis/BaruahF07,DBLP:conf/icdcn/BaruahF08,DBLP:conf/rtss/GuanSYY09,DBLP:conf/rtcsa/SunLA014,DBLP:conf/rtns/HuangC15}.
Baker~\cite{baker2006analysis} designed a test based on certain
properties to characterize a \emph{problem window}.  Baruah and
Fisher~\cite{DBLP:conf/opodis/BaruahF07,DBLP:conf/icdcn/BaruahF08}
used different annotations to extend the analysis window and derived
corresponding exponential-time
schedulability tests.  The first worst-case response-time analysis for
arbitrary-deadline task systems was proposed by Guan et al.
\cite{DBLP:conf/rtss/GuanSYY09}, where the authors used the insight
proposed by Baruah~\cite{baruah2007techniques} to limit the number of
carry-in jobs, and then apply the workload function proposed by
Bertogna et al.~\cite{DBLP:conf/opodis/BertognaCL05} to quantify the
requested demand of higher-priority tasks.  Unfortunately, it has
recently been shown by Sun et al.~\cite{DBLP:conf/rtcsa/SunLA014} that
this analysis in \cite{DBLP:conf/rtss/GuanSYY09}  is
optimistic. 
In addition, Sun et al.~\cite{DBLP:conf/rtcsa/SunLA014} derived a complex
carry-in workload function for the response time analysis where 
all possible combinations of carry-in and non-carry-in functions have to be
explicitly enumerated. However, their method is computationally intractable
since the time complexity is exponential.  Huang and Chen
\cite{DBLP:conf/rtns/HuangC15} proposed a more precise quantification
for the number of carry-in jobs of a 
task than the bounds used in the  tests provided 
in \cite{baker2006analysis,DBLP:conf/icdcn/BaruahF08}. They also presented a
response time bound for arbitrary-deadline tasks under
global scheduling in multiprocessor systems with linear-time complexity.

\subsection{Our Contribution}

We consider arbitrary-deadline sporadic task systems, which is the
most general case of the sporadic real-time task model.
To quantify the performance loss due to efficient schedulability tests
and the non-optimality of scheduling algorithms, we will adopt the
notion of speedup factors/bounds, also known as resource augmentation
factors/bounds. Table~\ref{tab:summary} summarizes the
state-of-the-art speedup bounds for the most adopted global fixed-priority scheduling algorithm, i.e., global deadline-monotonic (DM) scheduling.
Under global DM, a task $\tau_i$ has  higher priority than task
$\tau_j$ if $D_i \leq D_j$, in which ties are broken arbitrarily.
The authors note that the proof by Lundberg
 \cite{DBLP:conf/rtas/Lundberg02} seems incomplete. However, the
 concrete task set in \cite{DBLP:conf/rtas/Lundberg02}
 provides the lower bound $2.668$ of the speedup factors for global
 DM. Moreover, Andersson \cite{DBLP:conf/opodis/Andersson08}
showed that global slack monotonic scheduling has a speedup bound of
$\frac{3+\sqrt{5}}{2}\approx 2.6181$ for implicit-deadline task
systems. However, no better global fixed-priority scheduling
algorithms with respect to speedup factors are known for
constrained-deadline and arbitrary-deadline task systems.

{\bf Our Contributions:} Table~\ref{tab:summary} summarizes the
related results and the contribution of this paper for multiprocessor
global fixed-priority preemptive scheduling. We improve the best known
results by Baruah and Fisher \cite{DBLP:conf/opodis/BaruahF07} with
respect to the speedup bounds. Our contributions are:
\begin{itemize}
\item For \emph{any} global fixed-priority preemptive scheduling, we provide
  a series of schedulability tests 
  with different tradeoffs between time
  complexity and accuracy  in
  Section~\ref{sec:our-tests} and Section~\ref{sec:our-efficient-tests}.

\item We show that the global deadline-monotonic scheduling algorithm has a
  speedup factor $3-1/M$ with respect to the optimal
  multiprocessor scheduling policies when considering task systems
  with arbitrary deadlines. 
  This improves the analyses by Fisher
  and Baruah with respect to the
  speedup bounds, i.e., $4-1/M$ \cite{DBLP:conf/icdcn/BaruahF08} and
  $3.73$ \cite{DBLP:conf/opodis/BaruahF07}. 
  \item We show that all the schedulability
  tests we provide in this paper analytically dominate the tests by
  Baruah and Fisher \cite{DBLP:conf/opodis/BaruahF07} for global
  DM. We also show that global DM has a speedup lower bound of
  $3-3/(M+1)$, which shows that our schedulability analyses are
  asymptotically tight with respect to the speedup factors.
\end{itemize}

\begin{table*}[t]
  \centering
\renewcommand{\arraystretch}{1.3}
\scalebox{0.7}{
  \begin{tabular}{|l|p{2cm}|p{3.5cm}|p{5.3cm}|l|p{4.5cm}|}
    \hline
    & &implicit deadlines & constrained deadlines& arbitrary deadlines \\
    \hline
    \multirow{3}{*}{Global DM} & \multirow{2}{*}{upper bounds}& 2.668 \cite{DBLP:conf/rtas/Lundberg02} (poly.-time)& $3-1/M$ \cite{DBLP:journals/rts/BaruahBMS10} (expo.-time) & $\frac{2(M-1)}{4M-1-\sqrt{12M^2-8M+1}}\leq 3.73$ \cite{DBLP:conf/opodis/BaruahF07} (expo.-time)\\
    & & 2.823 \cite{DBLP:journals/corr/abs-k2q} (poly.-time) & $3-1/M$ \cite{DBLP:journals/corr/abs-k2q} (poly.-time) & $3-\frac{1}{M}$ (\emph{this paper}) (poly.-time)\\
    \cline{2-5}
    \cline{2-5}
    &  \multirow{2}{*}{lower bounds} & 2.668
    \cite{DBLP:conf/rtas/Lundberg02} & 2.668
    \cite{DBLP:conf/rtas/Lundberg02} &2.668 \cite{DBLP:conf/rtas/Lundberg02}\\
    & & & & $3- \frac{3}{M+1}$ (this paper)\\
    \hline
  \end{tabular}}
  \caption{Speedup bounds of the global deadline-monotonic (DM)
  scheduling algorithm for sporadic task systems. }
  \label{tab:summary}
\end{table*}


\section{System Model, Definitions, and Assumptions}
\label{sec:model}


We consider an arbitrary-deadline sporadic
task set ${\bf T}$ with $N$ tasks executed on $M \geq 2$ identical
processors based on global fixed-priority preemptive scheduling. We assume that
the priority levels of the tasks are unique (and given) and that 
$\tau_i$ has
higher priority than task $\tau_j$ if $i < j$. When there is only one processor, i.e., $M=1$, 
the existing results discussed in Section~\ref{sec:related-work} can be adopted, and
 our analysis here cannot be applied. We will implicitly use
the assumption $M \geq 2$ in the paper.

By definition, $M$ is an integer. 
We will implicitly assume that $D_i > 0$, $C_i > 0$, $T_i > 0$, $C_i/D_i \leq 1$, and $U_i \leq
1$  $\forall \tau_i$  
in this paper. Moreover, \emph{intra-task parallelism}
is not allowed. \emph{At most one job} of task $\tau_i$ can be
executed on at most one processor at each instant in time,
\emph{regardless of the number of the jobs of task $\tau_i$ awaiting
  for execution and the number of idle processors}. 
We denote the set of natural numbers as $\mathbb{N}$.

\subsection{Resource Augmentation}



We assume the original platform
speed is $1$. Therefore, running the platform at speed $s$ implies that the
worst-case execution time of task $\tau_i$ becomes $C_i/s$.
A scheduling algorithm $\mathcal{A}$ has a \emph{speedup bound} $s$
with respect to the optimal schedule, if it guarantees to always
produce a feasible solution when 1)~each processor is sped up to run
at $s$ times of the original speed of the platform and 2)~the task set
${\bf T}$ can be feasibly scheduled on the original $M$ identical
processors, i.e., running at speed 1.

We will use the negation of the above definition to quantify the
failure of algorithm $\mathcal{A}$: \emph{If $\mathcal{A}$ fails to
  ensure that all the task in ${\bf T}$ meet their deadlines, then no
  feasible multiprocessor schedule exists when each processor is
  slowed down to run at speed $1/s$.}



\subsection{Definitions and Necessary Condition}

We define the
following notation according to the task system and the priority
assignment:
\begin{itemize}
\item density $\delta_i$ of task $\tau_i$: $\delta_i = C_i/\min\{D_i, T_i\}$ 
\item maximum density $\delta_{\max}(k)$ among the first $k$ tasks:
  $\delta_{\max}(k)=\max_{i=1}^{k} \delta_i$ 
\item maximum between the utilization of the higher-priority tasks and the density
  of task \\$\tau_k$: $U_{\delta,k}^{\max}=\max\{\max_{i=1}^{k-1} U_i,
  \delta_k\}$ 
\item demand bound function \cite{Baruah90} $\dbf(\tau_i, t)$ of task $\tau_i$, further
  explained in Definition.~\ref{def:dbf-function} 
\item load $\load(k)$ of the first $k$ tasks: $\load(k)
  = \max_{t > 0} \frac{\sum_{i=1}^{k} \dbf(\tau_i, t)}{t}$ 
\end{itemize}

\begin{definition}[demand bound function ($\dbf$) by Baruah \cite{Baruah90}]
  \label{def:dbf-function}
  For any $t \geq 0$ 
  \begin{equation}
    \label{eq:dbf-function}
\dbf(\tau_i, t) = \max\left\{0,
    \left(\floor{\frac{t-D_i}{T_i}}+1\right)C_i\right\}
  \end{equation}
The demand bound function $\dbf(\tau_i, t)$ defines the execution time task
$\tau_i$ must finish for any interval length $t$ to ensure its timing correctness.
  \hfill\myendproof
\end{definition}

Since $\delta_i \geq U_i$ by definition,  we know that
$U_{\delta,k}^{\max} \leq \delta_{\max}(k)$. As we assume $C_i/D_i \leq 1$ and $U_i \leq
1$ we know that $\delta_i \leq 1$.
In addition to DBFs, we will heavily use the following workload function:
\begin{definition}[Workload function]
  Let $work_i(t)$ be a workload function, 
  representing the maximum amount of time for \emph{sequentially} executing
  the jobs of task $\tau_i$ released in time interval $[a,
  a+t)$, i.e., jobs released before $a$ are not considered. For
  any $t \geq 0$ 
  \begin{equation}
    \label{eq:workload-function}
    work_i(t) = \floor{\frac{t}{T_i}}C_i + \min\left\{C_i,t- \floor{\frac{t}{T_i}}T_i\right\}.
  \end{equation}
  For notational brevity, we set $work_i(t)$ to $-\infty$ if
  $t < 0$.  \hfill\myendproof
\end{definition}
The workload function $work_i(t)$  defined above is a piecewise function, i.e.,
linear in intervals $[\ell T_i, \ell T_i + C_i]$ with a slope $1$ and
constant, $(\ell+1)C_i$, 
in intervals $[\ell T_i + C_i, (\ell+1)T_i]$ for any non-negative integer $\ell$.
Two examples of the workload function are
 illustrated in Figure~\ref{fig:work-light-functions} in
 Section~\ref{sec:our-tests}.
To prove the speedup bound, 
we will utilize the following necessary condition. 
\begin{lemma}
  \label{lemma:lower-speed-bound}
  A task set ${\bf T}$ with $N$ tasks is not schedulable by any multiprocessor scheduling
  algorithm when the $M$ processors are running at any speed $s$, if 
  \begin{equation}
    \label{eq:def-s-overall}
    \max\left\{\max_{t> 0}\frac{\sum_{\tau_i \in {\bf T}} \dbf(\tau_i, t)}{Mt}, \frac{\sum_{\tau_i \in {\bf T}} U_i}{M}, \delta_{\max}(N) \right\} > s.    
  \end{equation}
\end{lemma}
\begin{proof}
  This is widely used based on a reformulation in the literature, e.g., \cite{DBLP:conf/opodis/BaruahF07,DBLP:conf/icdcn/BaruahF08}.
\end{proof}




\subsection{Analysis Based on DBFs}
Baruah and Fisher in~\cite{DBLP:conf/opodis/BaruahF07} provided a schedulability
test for task $\tau_k$ under global deadline-monotonic (DM) scheduling  that is
based on the Demand Bound Functions (DBF), assuming that the tasks are sorted
according to DM order already, i.e., $D_1 \leq D_2 \leq \ldots \leq D_N$: 

\begin{theorem}[Baruah and Fisher \cite{DBLP:conf/opodis/BaruahF07},
  revised in \cite{report-chen-baruah-errata-arbitrary-deadline-2007}]
\label{thm:baruah-fisher07-bug-free} Let $\mu_k$ be  defined as
$M-(M-1)\delta_{\max}(k)$. Task $\tau_k$ is schedulable under global DM if \footnote{The
  original proof by Baruah and Fisher \cite{DBLP:conf/opodis/BaruahF07} had
  a mathematical flaw in their Lemma 3, i.e., setting $\mu_k$ to
  $M-(M-1)\delta_k$. 
  It can be fixed by setting
  $\mu_k$ to $M-(M-1)\delta_{\max}(k)$.}
\begin{equation}
  \label{eq:sch-test-corollary-v1}
2\load(k) +
(\ceiling{\mu_k}-1)\delta_{\max}(k) \leq \mu_k.
\end{equation}  
\end{theorem}

\section{Schedulability Test by Pushing Forward}
\label{sec:our-tests}

In this section, we provide several 
conditions for the schedulability
of task $\tau_k$ under a given preemptive global fixed-priority scheduling
algorithm. They lead to a sufficient schedulability test for $\tau_k$, assuming
that the schedulability of
the tasks $\tau_1, \tau_2, \ldots, \tau_{k-1}$
under the given algorithm  is already  verified. 
This means that 
for all tasks $\tau_i$ with $i < k$ the worst-case
response time  is at most $D_i$. 
Therefore, the test should be applied for all tasks, i.e., from the
highest-priority task to the lowest-priority task, to ensure the schedulability of the task
set under the (specified/given) global fixed-priority scheduling. As the test
presented here has a high time complexity, we provide more efficient tests in
Section~\ref{sec:our-efficient-tests}. 

\subsection{Analysis Window Extension} 
We analyze the schedulability of $\tau_k$
by looking at the intervals where $\tau_k$ is active in the schedule 
$S$ provided by the global fixed-priority scheduling algorithm
according to the following definition:


\begin{definition}[{\bf active task}]
  For a schedule $S$, a task $\tau_i$ is active at time $t$, 
  if there is (at least) one job of $\tau_i$ that has arrived 
  before or at $t$ and has not finished yet at time $t$. \hfill\myendproof
\end{definition}


The schedulability conditions are proved by using \emph{contrapositive}.
Suppose 
a schedule $S$ produced
by the given global fixed-priority scheduling algorithm and that $t_d$
is the earliest (absolute) deadline at which a job of task $\tau_k$
misses its deadline.
Let $t_a$ be the time instant in $S$ such that
$\tau_k$ is continuously active in the time interval $[t_a, t_d)$
and 
is not active \emph{immediately} prior to $t_a$. By definition, $t_a$
must be the arrival time of a job of task $\tau_k$.  Suppose that
$t_d$ is the absolute deadline of the $\ell$-th job of task $\tau_k$ that
arrived in the time interval $[t_a, t_d)$. Therefore, as $ \tau_k$ is a sporadic
task, 
 $t_d - t_a \geq (\ell-1)T_k + D_k$. 
For notational brevity, we define $D_k'=(\ell-1) T_k + D_k$ and
$C_k'=\ell C_k$.

We remove all the jobs of task $\tau_k$ that
arrive before $t_a$ and all the jobs with priorities
lower than $\tau_k$ from  the schedule $S$. The
schedule of task $\tau_k$ remains unchanged in the resulting (new) schedule $S$,
due to the preemptiveness of the global fixed-priority scheduling algorithm. 
Let $C_k^*$ be the amount of time that task $\tau_k$ is executed from
$t_a$ to $t_d$. Since the $\ell$-th job of task $\tau_k$ misses its
deadline, we know that $C_k^* < \ell C_k = C_k'$. We now
introduce three functions that are defined for any $t \leq t_d$.
\begin{itemize}
\item Let $E(t, t_d)$ be the amount of workload (sum of the execution
  times) of the higher-priority jobs, i.e., from $\tau_1, \tau_2,
  \ldots, \tau_{k-1}$, \emph{executed} in the time interval $[t, t_d)$
  in schedule $S$.
\item Let $W(t, t_d)$ be $C_k^* + E(t, t_d)$.
\item Let $\Omega(t, t_d)$ be $\frac{W(t, t_d)}{t_d-t}$.
\end{itemize}

Those definitions and the deadline miss of task $\tau_k$ at time
$t_d$ lead to the following lemma.
\begin{lemma} Since $\tau_k$ misses its deadline at $t_d$ in $S$, the following
conditions hold: 
  \label{lemma:workload-overall}
  \begin{align}
    E(t_a, t_d)  \geq \;& M\times(t_d-t_a-C_k^*)\label{eq:E-geq-M}\\
    W(t_a, t_d)  > \;& M\times(t_d-t_a) - (M-1)C_k' \label{eq:W-geq-M}\\
    \Omega(t_a, t_d) > \;& M - (M-1)\times\frac{C_k'}{D_k'} \label{eq:Omega-geq-M}
  \end{align}
\end{lemma}
\begin{proof}
  Since task $\tau_k$ is active from $t_a$ to $t_d$ and is only
  executed for \emph{exactly} $C_k^*$ amount of time, we know that all
  $M$ processors must be busy executing other higher-priority jobs for
  at least $t_d-t_a-C_k^*$ amount of time. Therefore, the amount of
  workload $E(t_a, t_d)$ of the higher-priority jobs executed in the
  time interval $[t_a, t_d)$ must be at least $M\times (t_d-t_a-C_k^*)$, i.e., 
  Eq.~\eqref{eq:E-geq-M} must hold.\footnote{The condition in
    Eq.~\eqref{eq:E-geq-M} is widely used in the form
    of $E(t_a, t_d)  > M\times(t_d-t_a-\ell C_k)$. Here, since we will
    use $C_k^*$, the correct form is with $\geq$. }
Therefore, since $W(t_a, t_d)$ is defined as $E(t_a, t_d) + C_k^*$, we have
  \[
  W(t_a, t_d)  \geq M\times(t_d-t_a-C_k^*) +C_k^* > M\times(t_d-t_a) - (M-1)C_k',
  \]
  where the last inequality is due to $M \geq 2$ and $C_k' >
  C_k^*$. This leads to the conditions in Eq.~\eqref{eq:W-geq-M}.
  Since $\Omega(t_a, t_d)$ is defined as $\frac{W(t_a, t_d)}{t_d-t_a}$
  and $D_k' \leq t_d-t_a$, we have
\[
\Omega(t_a, t_d) \geq  M - (M-1)\frac{C_k'}{t_d-t_a} \geq M - (M-1)\frac{C_k'}{D_k'},  
\]
 i.e., the condition in Eq.~\eqref{eq:Omega-geq-M}.
\end{proof}

Although the interval $[t_a, t_d)$ can already be used for constructing
the schedulability tests, researchers  have tried
to push the interval of interest towards $[t_0, t_d)$ for some $t_0
\leq t_a$ based on certain properties, 
e.g., \cite{DBLP:conf/rtns/HuangC15,DBLP:conf/icdcn/BaruahF08,DBLP:conf/opodis/BaruahF07}. 
Such extensions have been
shown to provide better quantifications of the interfering
workload from the higher-priority tasks. In our analysis, we will
use a similar extension  strategy as suggested by Baruah and
Fisher~\cite{DBLP:conf/opodis/BaruahF07} based on a
user-specified parameter $\rho$. 

\begin{figure}[t]\centering
\scalebox{0.8}{
\def\ux{0.7cm} 
\def\uy{0.8cm} 
\begin{tikzpicture}[x=\ux,y=\uy,font=\sffamily,auto]
       \hspace{-0.5cm}

        \draw(24,0.3)node[anchor=south,red]{\Lightning};
        \draw(24,0.3)node[anchor=south west,red]{\footnotesize deadline miss};

        \draw[thick,->](0,0) -- (25, 0) node[anchor=west]{time};
        \draw[thick,->](24,1) -- (24, 0) node[anchor=north]{$t_d$};
        \draw[thick,->](15,0) node[anchor=north]{$t_a$} -- (15, 1);
        \draw[<->](15,0.3) -- (19.5, 0.3) node[anchor=south]{$\geq (\ell-1)T_k + D_k=D_k'$} -- (24, 0.3);
        \draw[thin, dotted] (7,0) node[anchor=north]{$t_0$} -- (7,1.3);
        \draw[thick,->](3,0) node[anchor=north]{$t_i$} -- (3, 1);
        \draw[thick,<-](11,0) node[anchor=north]{} -- (11, 1);
        \draw[<->](3,0.3) -- (7, 0.3) node[anchor=south]{$\tau_i$ is active} -- (11, 0.3);
        \draw[thick,<->](3,1.1) -- (5, 1.1) node[anchor=south]{$\phi_i$} -- (7, 1.1);
        \draw[thick,<->](7,1.1) -- (15.5, 1.1) node[anchor=south]{$\Delta$} -- (24, 1.1);

        \foreach \x in {3,7,15,24}
        \draw (\x,1pt) -- (\x,-3pt) ;
  
\end{tikzpicture}}
\caption{The notation used in Section~\ref{sec:our-tests}: 1) task
  $\tau_k$ is continuously active from $t_a$ to $t_d$ with a deadline
  miss at time $t_d$; 2) time instant $t_0$ is the smallest value of
  $t \leq t_a$ such that $\Omega(t, t_d) \geq \mu_k$; 3) time instant $t_i$ is the
  arrival time of a higher-priority carry-in task $\tau_i$ if $\tau_i$ is continuously
  active in time interval $[t_i, t_0+\varepsilon]$, where $t_i < t_0$ and
  $\varepsilon > 0$ is an arbitrarily small number; 4) $\phi_i$ is $t_0 -
  t_i$ and $\Delta$ is $t_d - t_0$. }
\label{fig:schedule-notation}
\end{figure}
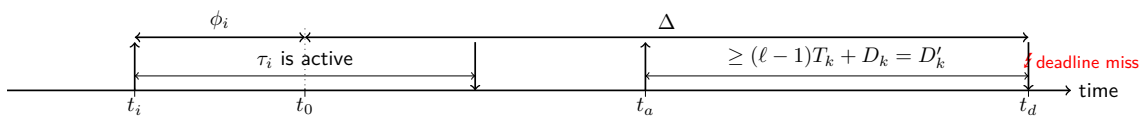

The following definition and lemmas are from
\cite{DBLP:conf/opodis/BaruahF07}. Figure~\ref{fig:schedule-notation}
provides an illustration of our notation based on the above
definitions.

\begin{definition}
  \label{definition-t0}
  Suppose that $\mu_k = M - (M-1)\rho$ for a certain $\rho$ with $1
  \geq \rho \geq \frac{C_k'}{D_k'}$. For the schedule $S$, let time
  instant $t_0$ be the smallest value of $t \leq t_a$ such that
  $\Omega(t, t_d) \geq \mu_k$. This means, $\Omega(t, t_d) < \mu_k$ for
  any $t < t_0$.
 \hfill\myendproof
\end{definition}

\begin{lemma}
  \label{lemma:existence-t0}
  If $\tau_k$ misses its deadline at $t_d$, for any $\rho$ with $1
  \geq \rho \geq \frac{C_k'}{D_k'}$, the time
  $t_0$, as defined in Definition~\ref{definition-t0}, always exists with
  $\Omega(t_0, t_d) \geq \mu_k$ and $t_0 \leq t_a$.
\end{lemma}
\begin{proof}
  By Eq.~\eqref{eq:Omega-geq-M} from Lemma~\ref{lemma:workload-overall} and $\rho \geq
  \frac{C_k'}{D_k'}$, we know 
\[ \Omega(t_a, t_d) > M -
  (M-1)\times\frac{C_k'}{D_k'} \geq M-(M-1)\rho=\mu_k.
\] Therefore,
  such a time instant $t_0 \leq t_a$ exists, at least 
  when the system starts.  
\end{proof}

\begin{definition}[{\bf carry-in task}]
  A task $\tau_i$ is a carry-in task in the schedule $S$, if 
  $\tau_i$ is continuously active in a time interval $[t_i, t_0+\varepsilon]$,
  for $t_i < t_0$ and an arbitrarily small $\varepsilon > 0$.
  \hfill\myendproof
\end{definition}


\begin{lemma}
  \label{lemma:number-carry-in-tasks}
  For  $1 \geq \rho \geq \frac{C_k'}{D_k'}$, there are at most
  $\ceiling{M - (M-1)\rho}-1$ carry-in tasks at $t_0$ in schedule
  $S$.
\end{lemma}

\subsection{Analysis Based on Workload Functions}

By extending the interval of interest to $[t_0, t_d)$,
Baruah and Fisher
provided the schedulability test  shown in
Theorem~\ref{thm:baruah-fisher07-bug-free} in this paper.  However,
they analyzed the workload in 
$[t_0, t_d)$ based on
the DBFs  by using the function
$\load(k)$ as an approximation, which will be shown pessimistic in
Corollary~\ref{cor:superior-to-BaruahFisher2007} in
Section~\ref{sec:global-DM}.  Moreover, their final analysis can only
be applied for global DM.
We will
carefully analyze the workload executed in 
$[t_0, t_d)$ 
to ensure that the analytical accuracy is 
 better preserved
and that the analysis can be used for any global fixed-priority preemptive
scheduling.
We will demonstrate that our analysis dominates the analysis by Baruah and
Fisher~\cite{DBLP:conf/opodis/BaruahF07}  in
Corollary~\ref{cor:superior-to-BaruahFisher2007}. 

\emph{For the analysis before
Theorem~\ref{thm:schedulability-test-main}, we will assume that
$\rho$ is given and $t_0$ is already defined. }
According to Lemma~\ref{lemma:number-carry-in-tasks}, at time $t_0$ at most
$\ceiling{M - (M-1)\rho}-1$ tasks are active  in schedule $S$. 
We quantify their contribution to the \emph{executed} workload in time
interval $[t_0,t_d)$ with two different forms from Lemma~\ref{lemma:workload-carry-in-tasks}, denoted by
$\omega^{heavy}_i(t_d-t_0)$,
and from Lemma~\ref{lemma:workload-carry-in-light-tasks}, denoted by
$\omega^{light}_i(t_d-t_0)$. While
Lemma~\ref{lemma:workload-carry-in-tasks} can be used in general,
Lemma~\ref{lemma:workload-carry-in-light-tasks} only holds if $U_i\leq \rho$. 


\begin{lemma}
  \label{lemma:workload-carry-in-tasks}
  If all jobs of a higher-priority task $\tau_i$ meet their deadlines, the
  upper bound $\omega^{heavy}_i(\Delta)$ on the workload
  of task $\tau_i$ executed from $t_0$ to
  $t_d$ with $\Delta=t_d-t_0$ in schedule $S$ is at most: 
  \begin{align}
    \label{eq:workload-heavy}
    \omega^{heavy}_i(\Delta) = work_i(\Delta + D_i).
  \end{align}
\end{lemma}
\begin{proof}
  Since all jobs of $\tau_i$ meet their deadlines,
  the jobs of  $\tau_i$ executed in $[t_0, t_d)$ must arrive in the time interval $(t_0
  - D_i, t_d)$. Therefore, the workload of task $\tau_i$ that can be
  sequentially executed is upper bounded by the workload function with
  length $t_d-(t_0-D_i)=\Delta+D_i$. 
\end{proof}


The key improvement achieved in this paper is due to the following
Lemma~\ref{lemma:workload-carry-in-light-tasks} to safely bound the
workload of a light task. 

Figure~\ref{fig:work-light-functions}
demonstrates the workload function for different cases in
Lemma~\ref{lemma:workload-carry-in-light-tasks}, together with a
linear approximation that will be presented in
Lemma~\ref{lemma:workload-carry-in-light-tasks-approximate}.  For the
workload function defined in Eq.~(\ref{eq:workload-light}), informally
speaking, the workload defined by $(p_2+1)C_i + \max\{0, C_i-\rho (T_i
- q_2)\}$ can be imagined as if 1) there is an offset for $C_i$ amount
of execution time at beginning of the interval, and 2) the workload in
each period starting from $C_i+p_2T_i$ to $C_i+(p_2+1)T_i$ is pushed
to the end of the period with a slope $\rho$. For example, in
Figure~\ref{fig:work-light-functions}(b), the offset is $3$, the
workload increases from $3$ at time $7$ to $6$ at time $13$ with a
slope $\rho=0.5$, the workload increases from $6$ at time $17$ to $9$
at time $23$ with a slope $\rho=0.5$, etc.

\begin{figure*}[t]\centering
  \subfloat[$U_i=0.3$ and $\rho=0.3$]{
  \centering\scalebox{0.7}{  
  \begin{tikzpicture}[node distance=0cm, y=0.24cm, x=0.14cm]
    \draw[->] (0,0) -- coordinate (x axis mid) (64,0) node[right]{$\Delta$};
    \draw[->] (0,0) -- coordinate (y axis mid) (0,21) node[above]{};
    \foreach \x in {0,5,...,60}
    \draw [line width=0.02pt, dotted, gray](\x,20) -- (\x,-3pt) 
    node[anchor=north] {\x}; 
    \foreach \x in {0,1,...,60}
    \draw [line width=0.02pt, dotted, gray](\x,20) -- (\x,-3pt) node[]{};
    \foreach \y in {0,2,...,20}
    \draw [line width=0.02pt, dotted, gray] (60, \y) -- (-3pt, \y) 
    node[anchor=east] {\y}; 
   \foreach \y in {0,1,...,20}
    \draw [line width=0.02pt, dotted, gray] (60, \y) -- (-3pt, \y) 
    node[anchor=east] {}; 
    
		\draw[line width = 1pt,color=red, dotted] plot[]
     (0, 0) -- (3, 3) -- (10, 3) -- (13, 6) -- (20, 6) -- (23, 9) -- (30, 9)
     -- (33, 12) -- (40, 12) -- (43, 15) -- (50, 15) -- (53, 18) -- (60, 18);

	\draw[color= black, line width = 1.5pt]
        (3,3) -- (53,18) -- (60, 20.1);

	\draw[color= blue, line width = 3pt, dashed]
        (0, 2.1) -- (60,20.1) ;

      \draw[black](70.5,4) node[anchor=east]{\footnotesize
      $\omega_i^{light}(\Delta)=(p_2+1)C_i + \max\{0, C_i-\rho (T_i - q_2)\}$
      (solid)}; \draw[red](70.5,6)
      node[anchor=east]{\footnotesize $work_i(\Delta)$ (dotted)}; 
      \draw[blue](70.5,2) node[anchor=east]
      {\footnotesize safe approximation of $\omega_i^{light}(\Delta)$ in
      Lemma~\ref{lemma:workload-carry-in-light-tasks-approximate} (dashed)};
  \end{tikzpicture}}}
\subfloat[$U_i = 0.3$ and $\rho=0.5$]{
  \centering\scalebox{0.7}{  
  \begin{tikzpicture}[node distance=0cm, y=0.24cm, x=0.14cm]
    \draw[->] (0,0) -- coordinate (x axis mid) (64,0) node[right]{$\Delta$};
    \draw[->] (0,0) -- coordinate (y axis mid) (0,21) node[above]{};
    \foreach \x in {0,5,...,60}
    \draw [line width=0.02pt, dotted, gray](\x,20) -- (\x,-3pt) 
    node[anchor=north] {\x}; 
    \foreach \x in {0,1,...,60}
    \draw [line width=0.02pt, dotted, gray](\x,20) -- (\x,-3pt) node[]{};
    \foreach \y in {0,2,...,20}
    \draw [line width=0.02pt, dotted, gray] (60, \y) -- (-3pt, \y) 
    node[anchor=east] {\y}; 
   \foreach \y in {0,1,...,20}
    \draw [line width=0.02pt, dotted, gray] (60, \y) -- (-3pt, \y) 
    node[anchor=east] {}; 
    
		\draw[line width = 1pt,color=red, dotted] plot[]
     (0, 0) -- (3, 3) -- (10, 3) -- (13, 6) -- (20, 6) -- (23, 9) -- (30, 9)
     -- (33, 12) -- (40, 12) -- (43, 15) -- (50, 15) -- (53, 18) -- (60, 18);

	\draw[color= black, line width = 1.7pt]
        (3,3) -- (7,3) -- (13, 6) -- (17, 6) -- (23, 9) -- (27, 9)
     -- (33,12) -- (37,12) -- (43, 15) -- (47, 15) -- (53, 18) -- (57, 18) -- (60, 19.5);

	\draw[color= blue, line width = 2pt, dotted]
        (0, 2.1) -- (60,20.1) ;

      \draw[black](70.5,4) node[anchor=east]{\footnotesize
      $\omega_i^{light}(\Delta)=(p_2+1)C_i + \max\{0, C_i-\rho (T_i - q_2)\}$
      (solid)}; \draw[red](70.5,6)
      node[anchor=east]{\footnotesize $work_i(\Delta)$ (dotted)}; 
      \draw[blue](70.5,2) node[anchor=east]
      {\footnotesize safe approximation of $\omega_i^{light}(\Delta)$ in
      Lemma~\ref{lemma:workload-carry-in-light-tasks-approximate} (dashed)};

  \end{tikzpicture}}}
\caption{Two examples for the approximation of $work_i$ for $\tau_i$
with $T_i=10, C_i=3, D_i = 45$: 
black curves for
$\omega_i^{light}(\Delta)$ defined in
Lemma~\ref{lemma:workload-carry-in-light-tasks} 
and the 
approximation 
in Lemma~\ref{lemma:workload-carry-in-light-tasks-approximate} (blue curves).}
  \label{fig:work-light-functions} 
\end{figure*}
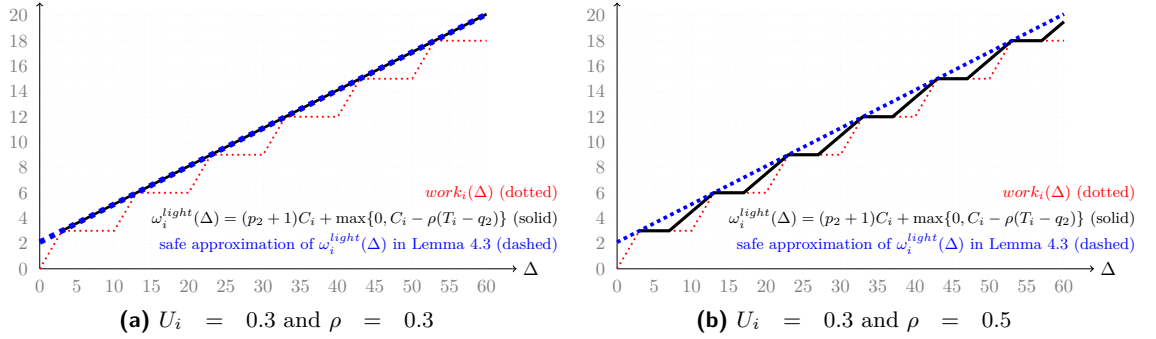

\begin{lemma}
  \label{lemma:workload-carry-in-light-tasks}
  If all jobs of a higher-priority task $\tau_i$ meet their deadlines
  and $U_i \leq
  \rho \leq 1$, the upper bound $\omega^{light}_i(\Delta)$ on the workload
  of task $\tau_i$ executed from $t_0$ to $t_d$ with $\Delta=t_d-t_0$ in schedule $S$ is:
  {\small
  \begin{align}
    \label{eq:workload-light}
    \omega^{light}_i(\Delta) =    
\begin{cases}
      \Delta & \mbox{ if } 0 < \Delta \leq
      C_i\\      
  \max
    \begin{Bmatrix*}[l]
   work_i(\Delta),\\
     (p_2+1)C_i + \max\{0,C_i-\rho (T_i-q_2)\}
    \end{Bmatrix*} 
& \mbox{ if } \Delta > C_i
    \end{cases}
  \end{align} }   
  where $p_2 = \ceiling{(\Delta-C_i)/T_i}-1$ and $q_2$ is $\Delta-C_i
  - p_2 T_i$.
\end{lemma}
\begin{proof}
  As the case when 
    $0 < \Delta \leq C_i$ is due to the definition, let $\Delta > C_i$
    for   the rest of the proof.
  Based on the
  schedule $S$, let $t_i< t_0$ be the time instant such that task
  $\tau_i$ is continuously active in the time interval $[t_i, t_0]$
  and task $\tau_i$ is not active \emph{immediately} prior to $t_i$. If $t_i$ does not
  exist, then task $\tau_i$ does not have 
  workload released
  before $t_0$ that is still active. Therefore, the worst-case workload is 
  $work_i(\Delta)$ in this case. 

  By the definition of $t_i$, if it exists, there are at
  most $\ceiling{\frac{\phi_i}{T_i}}$ jobs of task $\tau_i$ executed
  in time interval $(t_i, t_0]$.
  For the rest of the proof, we only consider that $t_i$ exists and that 
  $\Delta > C_i$.   By definition, $t_i$ must be the arrival time of a
  job of task $\tau_i$. Moreover, due to the definition of $t_0$ in
  Definition~\ref{definition-t0}, we know that $\Omega(t_i, t_d) <
  M-(M-1)\rho$. Let $\phi_i$ be $t_0 - t_i$. Since $\Omega(t_i, t_d) <
  M-(M-1)\rho$ and $\Omega(t_0, t_d) \geq M-(M-1)\rho$, we have
  \begin{align}
    W(t_0, t_d) =\Omega(t_0, t_d) \cdot (t_d-t_0) & \geq (t_d-t_0)
    \mu_k = \Delta\mu_k \label{eq:W-of-td-t0}\\
    W(t_i, t_d) = \Omega(t_i, t_d) \cdot (t_d-t_i)& < (t_d-t_i) \mu_k
    = (\Delta+\phi_i)\mu_k \label{eq:W-of-td-ti}
  \end{align}
  Substracting Eq.~(\ref{eq:W-of-td-ti}) by Eq.~(\ref{eq:W-of-td-t0}), we have $ W(t_i, t_d) -  W(t_0, t_d) < \phi_i \mu_k$, i.e., in schedule $S$ the workload executed in
  time interval $[t_i, t_0)$ is \emph{strictly less} than $\phi_i
  \mu_k$. Suppose that $y_i$ is the amount of time that task $\tau_i$
  is executed in time interval $[t_i, t_0)$, i.e., task $\tau_i$ is
  active but blocked by other higher-priority jobs for $\phi_i - y_i$
  amount of time in this time interval. 
  When task $\tau_i$
  is blocked in global fixed-priority scheduling,  
  all the $M$
  processors are executing other jobs. The workload executed in time
  interval $[t_i, t_0)$ is at least $M(\phi_i - y_i) + y_i$. 
  Therefore, by the above discussions, we know that
  \begin{equation}
    \label{eq:y-lowerbound}
M(\phi_i - y_i) + y_i < \phi_i \mu_k = \phi_i (M-(M-1)\rho)\Rightarrow y_i >
\rho \phi_i,
  \end{equation}
  since $M \geq 2$. At time $t_0$, the remaining execution time of the
  jobs of task $\tau_i$ that arrived before $t_0$ in schedule $S$ is at
  most $\ceiling{\phi_i/T_i}C_i - \rho\phi_i$. Note that the
  existence of $t_i$ in our definition means that
  $\ceiling{\phi_i/T_i}C_i - y_i > 0$, i.e., $\ceiling{\phi_i/T_i}C_i - \rho\phi_i > 0$.

  The workload of task $\tau_i$ that is executed in the time interval $[t_i,
  t_d)$ in schedule $S$ is at most $work_i(t_d-t_i) =
  work_i(\Delta+\phi_i)$.  The workload of task $\tau_i$ that is executed in
  the time interval $[t_i, t_d)$ is at least $y > \rho \phi_i$. Therefore,
  the workload of task $\tau_i$ that is executed in the time interval $[t_0,
  t_d)$ in schedule $S$ is upper bounded by $work_i(\Delta+\phi_i) - \rho
  \phi_i$. 

  The rest of the proof is to provide an upper bound of
  $work_i(\Delta+\phi_i) - \rho \phi_i$ for any arbitrary $\phi_i >
  0$. The proof involves some detailed manipulations of the workload
  function. Before proceeding, we explain two basic properties of the
  workload function here:
  \begin{itemize}
  \item When $p$ is a non-negative integer and $0 \leq x\leq C_i$,
    $work_i(p T_i+x)=p C_i+x$.
  \item When $p$ is a non-negative integer and $0 \leq x$,
    $work_i(p T_i+x)=p C_i + work_i(x)$.
  \end{itemize}

  To identify the exact value of $work_i(\Delta+\phi_i)$, we define
  the following variables $p_1, p_2, q_1$, and $q_2$ for brevity:
  \begin{itemize}
  \item Let $p_1$ be $\ceiling{\phi_i/T_i}-1$ and $q_1$ be $\phi_i -
    p_1 T_i$, i.e., $p_1+1$ is the
    number of jobs of task $\tau_i$ that can be released in $[t_i,
    t_0]$. By definition $\phi_i > 0$, which implies that $p_1$ is a non-negative integer, $0 <
    q_1 \leq T_i$, and $\phi_i=p_1 T_i + q_1$.
  \item Let $p_2$ be $\ceiling{(\Delta-C_i)/T_i}-1$ and $q_2$ be
    $\Delta-C_i -  p_2 T_i$, i.e., $p_2+1$ is the
    number of jobs of task $\tau_i$ that can be released in $[t_0+C_i, t_d]$. Due to the assumption $\Delta > C_i$, we know that $p_2$ is a non-negative
    integer, $0 < q_2 \leq T_i$, and $\Delta-C_i = p_2 T_i + q_2$.
  \end{itemize}
  By the above definition, we achieve $\phi_i+\Delta = (p_1+p_2)T_i + q_1+q_2+ C_i$, and
  \begin{align}
& work_i(\Delta+\phi_i) - \rho \phi_i  \nonumber\\
&
=  work_i((p_1+p_2)T_i + q_1+q_2+ C_i) - \rho(p_1T_i + q_1)\nonumber\\
&=   work_i(p_2T_i + q_1+q_2+ C_i) + p_1 C_i - \rho(p_1T_i + q_1)\nonumber\\
&=   work_i(p_2T_i + q_1+q_2+ C_i) + p_1 U_i T_i - \rho(p_1T_i + q_1)\nonumber\\
&\leq   work_i(p_2 T_i + q_1+q_2+ C_i) - \rho q_1 \label{eq:pq1q2}
  \end{align}
  where the inequality is due to the assumption that  $0 \leq U_i \leq \rho$.
  We will prove that the right-hand side of
  Eq.~\eqref{eq:workload-light} is a safe upper bound on the condition
  in Eq.~\eqref{eq:pq1q2}. By the definition of $q_1$ and $q_2$, we
  know that $0 \leq q_1+q_2 \leq 2 T_i$, i.e., $C_i \leq p_2 T_i + q_1+q_2+ C_i
  \leq 2T_i+C_i$.  Depending on the value of
  $q_1+q_2$, there are  four cases for different (linear or constant)
  segments of $work_i(p_2 T_i + q_1+q_2+ C_i)$  to be analyzed:
  \begin{itemize}
  \item {\bf Case 1:} $0 \leq q_1+q_2 \leq T_i - C_i$: That is, 
$p_2T_i + C_i \leq p_2 T_i + q_1+q_2+ C_i \leq p_2 T_i + T_i$. Therefore, 
   $work_i(p_2 T_i + C_i) \leq work_i(p_2 T_i + q_1+q_2+ C_i) \leq
   work_i(p_2 T_i + T_i)$. Since $work_i(p_2 T_i + C_i) = work_i(p_2
   T_i + T_i) = (p_2+1)C_i$, we have
    \begin{align*}
     \mbox{RHS. of Eq.~\eqref{eq:pq1q2}} = &  (p_2+1)C_i - \rho q_1 
\leq work_i (p_2T_i+C_i+q_2)= work_i(\Delta),
    \end{align*}
    where $\leq$ is due to $\rho \geq 0$ and $q_1 > 0$.
  \item \textbf{Case 2:} $T_i - C_i < q_1+q_2 \leq T_i$:
 By definition, when $p_2$ is a nonnegative integer and $0<x\leq C_i$, $work_i((p_2+1)T_i+x)=(p_2+1)C_i+x$. By $T_i-C_i<q_1+q_2\leq T_i$, we know that $(p_2+1)T_i<p_2T_i+q_1+q_2+C_i\leq (p_2+1)T_i+C_i$. Therefore, $work_i(p_2T_i+q_1+q_2+C_i)=(p_2+1)C_i+(p_2T_i+q_1+q_2+C_i-(p_2+1)T_i) =  (p_2+1)C_i+(q_1+q_2+C_i-T_i)$. Let $\eta$ be
    $T_i-(q_1+q_2)$. By definition $\eta \geq 0$.
Therefore,
    \begin{align*}
      \mbox{RHS. of Eq.~\eqref{eq:pq1q2}} = &  (p_2+1)C_i + (C_i - \eta) - \rho (T_i-q_2 - \eta) \nonumber\\
      = &  (p_2+1)C_i + (C_i - \rho (T_i-q_2)) + \eta(\rho-1) \nonumber\\
      \leq & (p_2+1)C_i + \max\{0, C_i-\rho (T_i - q_2)\},
    \end{align*}
    where $\leq$ is due to $0 \leq \rho \leq 1$ and $\eta\geq 0$.
\item \textbf{Case 3:} $T_i < q_1+q_2 \leq 2T_i-C_i$: Thus,
  $work_i(p_2 T_i + q_1+q_2+ C_i)=(p_2+2)C_i$, and
\[
     \mbox{RHS. of Eq.~\eqref{eq:pq1q2}} =  (p_2+1)C_i +C_i- \rho q_1 
\leq   (p_2+1)C_i + \max\{0, C_i-\rho (T_i - q_2)\},
\]
    where $\leq$ is due to $\rho \geq 0$ and $q_1+q_2 > T_i$.
\item \textbf{Case 4:} $2T_i - C_i < q_1+q_2 \leq 2T_i$: In this case
  $work_i(p_2 T_i + q_1+q_2+ C_i)$ is equal to $(p_2+2)C_i + (q_1+q_2 + C_i - 2T_i)$, similar to the analysis in Case 2. Let $\eta$ be
  $2T_i-(q_1+q_2)$. By definition $\eta \geq 0$.  Therefore,
   {\small \begin{align*}
      \mbox{RHS. of Eq.~\eqref{eq:pq1q2}} =&   (p_2+1)C_i + 2C_i - \eta - \rho (2T_i-q_2 - \eta) \\
      = &  (p_2+1)C_i + C_i + T_i(U_i - \rho) - \eta (1- \rho)- \rho (T_i-q_2) \\
      \leq & (p_2+1)C_i + \max\{0, C_i-\rho (T_i - q_2)\},
    \end{align*}}
    where $\leq$ is due to $0 < U_i \leq \frac{C_k'}{D_k'} \leq \rho \leq 1$ and $\eta \geq 0$, i.e., $U_i - \rho \leq 0$ and $- \eta (1- \rho) \leq 0$.
  \end{itemize}
  Since $0 < q_1+q_2 \leq 2 T_i$, we know that $work_i(\Delta)$ is a
  safe upper bound for $\textbf{Case 1}$ and that $(p_2+1)C_i + \max\{0, C_i
  -\rho (T_i - q_2)\}$ is a safe upper bound for the other cases, and we reach the
  conclusion of this lemma.  
\end{proof}

\begin{lemma}\label{lemma:heavy>light}
  $\forall \Delta > 0$, $\omega_i^{heavy}(\Delta) \geq \omega_i^{light}(\Delta)$.
\end{lemma}
\begin{proof}
  This inequality can be proved formally, but can also be derived by
  following the definitions. When $0 < \Delta \leq C_i$, the
  inequality holds naturally.  In the proof of
  Lemma~\ref{lemma:workload-carry-in-light-tasks}, \emph{the workload
    of task $\tau_i$ that is executed in the time interval $[t_i,
    t_d)$ in schedule $S$ is at most $work_i(t_d-t_i) =
    work_i(\Delta+\phi_i)$.} Since $\phi_i \leq D_i$, we know that $\omega_i^{light}(\Delta) \leq 
work_i(\Delta+\phi_i) \leq work_i(\Delta+D_i) =
    \omega_i^{heavy}(\Delta)$.
\end{proof}

Here is a short summary of the
information provided by
Lemmas~\ref{lemma:number-carry-in-tasks},~\ref{lemma:workload-carry-in-tasks},~and~\ref{lemma:workload-carry-in-light-tasks}.
\begin{compactitem}
\item According to Lemma~\ref{lemma:number-carry-in-tasks}, at time
  $t_0$, there are at most $\ceiling{M -
    (M-1)\rho}-1=\ceiling{\mu_k}-1$ carry-in tasks.
\item Among the $\ceiling{\mu_k}-1$ carry-in tasks, there are two
  types of carry-in tasks, i.e., \emph{heavy} and \emph{light} tasks. A
  light carry-in task $\tau_i$ can be described by $\omega_i^{light}(\Delta)$
  from Eq.~\eqref{eq:workload-light} if the utilization is no more
  than $\rho$ and a heavy carry-in task $\tau_i$ can be described by
  $\omega_i^{heavy}(\Delta)$ from Eq.~\eqref{eq:workload-heavy}. By
  observing the conditions in
  Eqs.~\eqref{eq:workload-heavy}~and~\eqref{eq:workload-light}, we
  know that $work_i(\Delta) \leq \omega_i^{light}(\Delta) \leq
  \omega_i^{heavy}(\Delta)$. 
\item Since $\rho$ is a user-defined parameter, a smaller $\rho$
  implies a larger $\mu_k$, i.e., potentially more carry-in tasks and
  more heavy carry-in tasks. By constrast, a larger $\rho$ implies a
  smaller $\mu_k$, i.e., potentially less carry-in tasks and more
  light carry-in tasks. Therefore, \emph{a larger $\rho$ is better for
    minimizing the carry-in workload}.
\item However, the window of interest $[t_0, t_d)$ is defined by the
  condition $\Omega(t_0, t_d) \geq M-(M-1)\rho$. The window of
  interest is smaller when $\rho$ is larger. As a result, there is no
  monotonicity with respect to the schedulability test for setting the
  value of $\rho$.
\end{compactitem}

\begin{theorem}
  \label{thm:schedulability-test-main}
  Task $\tau_k$ is schedulable by the given global fixed-priority scheduling if
 {\small \begin{align}
& \forall \ell \in \mathbb{N}, \exists 1 \geq \rho \geq
\ell C_k/((\ell-1) T_k + D_k), \forall \Delta \geq (\ell-1)
T_k+D_k\nonumber \\
& \ell C_k + \sum_{\tau_i \in {\bf T}^{carry}}
\omega^{diff}_i(\Delta, \rho) +
\sum_{i=1}^{k-1}work_i(\Delta) \leq \Delta \cdot \mu_k\label{eq:theorem-schedulability-condition}
  \end{align}}
holds, where $\mu_k = M - (M-1)\rho$,
\begin{equation}
  \label{eq:omega-diff}
  \omega^{diff}_i(\Delta, \rho) =
  \begin{cases}
    \omega^{heavy}_i(\Delta) -work_i(\Delta) & \mbox{ if } U_i > \rho\\
    \omega^{light}_i(\Delta) - work_i(\Delta) & \mbox{ if } U_i \leq \rho
  \end{cases}
\end{equation}
 and ${\bf T}^{carry}$ is the set of the $\ceiling{\mu_k}-1$ tasks among the $k-1$ higher-priority tasks with the largest values of $\omega^{diff}_i(\Delta, \rho)$. If $D_k \leq T_k$, we only need to
consider $\ell=1$.
\end{theorem}
\begin{proof}
  We prove this theorem by contrapositive, i.e., task $\tau_k$ misses
  its deadline first at time $t_d$ in a global fixed-priority
  preemptive schedule $S$. We know that $t_a$ can be defined for
  schedule $S$, and $t_0$, i.e., $\Omega(t_0, t_d) \geq M -
  (M-1)\times\frac{C_k'}{D_k'}$ in Definition~\ref{definition-t0} can
  be defined for any $\rho$ with $1 \geq \rho \geq \ell C_k/((\ell-1)
  T_k + D_k)$ due to Lemma~\ref{lemma:existence-t0}.

  By the
  existence of $t_d$, the choice of $\rho$, and the definition of
  $t_0$ in Definition~\ref{definition-t0}, we know that  the
  deadline miss of task $\tau_k$ at time $t_d$ in the schedule $S$  implies 
 \begin{align} 
      \exists \ell \in \mathbb{N}, \forall 1 \geq \rho \geq
      \ell C_k/((\ell-1) T_k + D_k), \exists \Delta=t_d-t_0, \qquad 
       \Omega(t_0, t_d) \geq M -
      (M-1)\rho\label{eq:necessary-miss-v1}
    \end{align}
  By the fact that $C_k^* < C_k' = \ell C_k$ and the definition of $\Omega()$, we have 
  \begin{align}
    \label{eq:omega-vs-mu}
    \Omega(t_0, t_d) = \frac{C_k^*+E(t_0, t_d)}{t_d-t_0}  <  \frac{\ell C_k + E(t_0, t_d) }{t_d-t_0}
  \end{align}

  By Lemma~\ref{lemma:number-carry-in-tasks}, for a specific $\rho$,
  there are at most $\ceiling{M - (M-1)\rho}-1=\ceiling{\mu_k}-1$
  higher-priority carry-in tasks at time $t_0$ and the other
  higher-priority tasks do not have any unfinished job at time
  $t_0$. Suppose that ${\bf T}^{heavy}$ and ${\bf T}^{light}$ are the
  sets of the heavy and light carry-in tasks at time $t_0$,
  respectively. By Lemma~\ref{lemma:number-carry-in-tasks}, $|{\bf
    T}^{heavy}| + |{\bf T}^{light}| \leq \ceiling{\mu_k}-1$.
 Therefore, by using
  Lemmas~\ref{lemma:workload-carry-in-tasks}~and~\ref{lemma:workload-carry-in-light-tasks}~and~\ref{lemma:heavy>light},
  we have
{\small \begin{align}
&  E(t_0, t_d)  \leq \sum_{\tau_i \in {\bf T}^{heavy}}
     \omega^{heavy}_i(\Delta)  + \sum_{\tau_i \in {\bf T}^{light}}
      \omega^{light}_i(\Delta) \nonumber\\
=&  \sum_{\tau_i \in {\bf T}^{heavy}}
      \left(\omega^{heavy}_i(\Delta)  - work_i(\Delta) \right) + \sum_{\tau_i \in {\bf T}^{light}}
      \left(\omega^{light}_i(\Delta)  - work_i(\Delta) \right) 
      \;\;+ \sum_{i=1}^{k-1} work_i(\Delta) \\
 \leq &   \sum_{\tau_i \in {\bf T}^{carry}}
\omega^{diff}_i(\Delta, \rho) +
\sum_{i=1}^{k-1}work_i(\Delta)  \label{eq:E-sum-with-diff}
\end{align}}
where $\omega_i^{diff}(\Delta,\rho)$ is defined in
Eq.~\eqref{eq:omega-diff}, and ${\bf T}^{carry}$ is defined in the statement of the theorem.

By Eqs.~\eqref{eq:necessary-miss-v1},~\eqref{eq:omega-vs-mu},~and~\eqref{eq:E-sum-with-diff},
and the fact $t_d-t_a \geq D_k' = (\ell-1) T_k+D_k$, the deadline miss
of task $\tau_k$ at $t_d$ implies
  \begin{align} & \exists \ell
    \in \mathbb{N}, \forall 1 \geq \rho \geq
    \ell C_k/((\ell-1) T_k + D_k), \exists \Delta \geq  (\ell-1) T_k+D_k\nonumber    \\
    &\ell C_k + \sum_{\tau_i \in {\bf T}^{carry}}
    \omega^{diff}_i(\Delta, \rho) + \sum_{i=1}^{k-1}work_i(\Delta) >
    \Delta \cdot \mu_k
  \end{align}
  
Therefore, the negation of the above necessary condition for the
deadline miss of task $\tau_k$ at time $t_d$ is a safe sufficient
schedulability test. We reach the conclusion of the schedulability test.

  When $D_k \leq T_k$, since $t_d$ is the earliest moment in the
  schedule $S$ with a deadline miss of task $\tau_k$, we know that
  $t_a$ is by definition $t_d - D_k$ and $\ell$ is $1$. Therefore, we only have to
  consider $\ell=1$ when $D_k \leq T_k$.
\end{proof}

The schedulability test described in
Theorem~\ref{thm:schedulability-test-main} can be informally explained
as follows: 1) it requires to test all the possible positive integers
for $\ell$, like the busy-window concept, 2) it has to find a $\rho$
value in the specified range, and 3) for the specified
combination of $\ell$ and $\rho$, we have to test whether the
condition in Eq.~\eqref{eq:theorem-schedulability-condition} holds for
every $\Delta \geq (\ell-1)T_k+D_k$.

\subsection{Remarks on Implementing
  Theorem~\ref{thm:schedulability-test-main}}
\label{sec:remark-of-exponential-time}

Unfortunately, due to the following issues, implementing the 
schedulability test in Theorem~\ref{thm:schedulability-test-main} directly would
lead to a 
high time complexity: 
\begin{itemize}
\item {\bf Issue 1 due to $\Delta$}: For specific $\ell$ and $\rho$,
  testing the schedulability condition in
  Eq.~\eqref{eq:theorem-schedulability-condition} requires to evaluate
  all $\Delta \geq (\ell-1)T_k+D_k$. Suppose that $HP(k)$ is the
  hyper-period of $\setof{\tau_1, \tau_2, \ldots, \tau_{k-1}}$, i.e., the least
  common multiple of the periods of 
  $\tau_1, \tau_2, \ldots, \tau_{k-1}$.  
  Since $work_i(\Delta) + HP(k) U_i=
  work_i(\Delta + HP(k))$, $\omega_i^{light}(\Delta) + HP(k) U_i=
  \omega_i^{light}(\Delta + HP(k))$, and $\omega_i^{heavy}(\Delta) + HP(k)
  U_i= \omega_i^{heavy}(\Delta + HP(k))$, we only have to 
  test $\Delta
  \in [(\ell-1)T_k+D_k, (\ell-1)T_k+D_k+ HP(k)]$, as long as
  $\sum_{i=1}^{k-1} U_i \leq \mu_k$. However, the time complexity can
  still be exponential.  
  We will explain how to reduce this
  complexity by using safe upper bounds in
  Section~\ref{sec:our-efficient-tests}.
  
\item {\bf Issue 2 due to $\rho$}: For a specific $\ell$, the
  schedulability condition in
  Eq.~\eqref{eq:theorem-schedulability-condition} is dependent on the
  selection of $\rho$. If $\rho$ is smaller, then $\mu_k$ is larger, 
  and vice versa. A smaller $\rho$ increases the right-hand side
  in the schedulability test in
  Eq.~\eqref{eq:theorem-schedulability-condition}, but it also increases 
  the left-hand side, since there are potentially more carry-in
  tasks. One simple strategy to find a suitable $\rho$ instead of
  searching for all values of $\rho$ is to start from $\rho=\ell C_k/((\ell-1)
  T_k + D_k$ and increase $\rho$ to the next (higher) $U_i$ for
  certain higher-priority task $\tau_i$ if necessary. Therefore, in
  the worst case, we only have to consider $k$ different $\rho$
  values. We will deal with this in
  Theorems~\ref{thm:schedulability-test-all-range}~and~\ref{thm:schedulability-test-ell=all-range-deltamax} in
  Section~\ref{sec:our-efficient-tests}.
\item {\bf Issue 3 due to $\ell$}: We need to consider all 
  positive integer  values of $\ell$ in the schedulability condition in
  Eq.~\eqref{eq:theorem-schedulability-condition}, as 
  the test
  is only valid when the condition holds for all $\ell \in
  \mathbb{N}$. Therefore, if we only test some $\ell$, it is necessary to
  prove that the other $\ell$ configurations are also covered even
  though they are not tested.  We will explain how to deal with this
  in
  Theorems~\ref{thm:schedulability-test-ell=1-or-infty}~and~\ref{thm:linear-time-pessimistic}
  in Section~\ref{sec:our-efficient-tests}.
\end{itemize}


\section{Efficient Schedulability Tests}
\label{sec:our-efficient-tests}

In this section we
provide several schedulability tests based on approximate workload
functions to test the schedulability of task $\tau_k$ more efficiently. The
following three lemmas approximate the \emph{piecewise linear} workload function $work_i(\Delta)$,
$\omega_i^{heavy}(\Delta)$ and $\omega_i^{light}(\Delta)$ by \emph{linear}
functions with respect to $\Delta$ for any $\Delta \geq 0$.

\begin{lemma}
    \label{lemma:workload-approximate}  
    When $0 \leq U_i \leq 1$, for any $\Delta \geq 0$,
    \begin{equation}
      \label{eq:workupper-linear}
  work_i(\Delta) \leq C_i - C_iU_i + U_i\Delta.      
    \end{equation}
\end{lemma}
\begin{proof}
 This inequality was already stated in Eq. (5) by Bini et
  al. \cite{DBLP:journals/tc/BiniNRB09} as a fact.  Here, we provide
  the proof for completeness.    Suppose that $\Delta$ is $p_3
  T_i + q_3$, where $p_3$ is $\floor{\frac{\Delta}{T_i}}$ and $q_3$ is
  $\Delta - \floor{\frac{\Delta}{T_i}}T_i$. Therefore, we know
  $U_i\Delta = p_3 C_i + q_3U_i$ and $work_i(\Delta) = p_3 C_i +
  \min\{C_i, q_3\}$. We have to consider two cases:
  \begin{itemize}
  \item If $q_3 \leq C_i$: we have 
{\small    \begin{align*}
work_i(\Delta)  = \;\;& p_3 C_i + q_3 \leq p_3 C_i + C_i -  (C_i - q_3) \\
\leq_1 \;\;& p_3 C_i + C_i - (C_i - q_3)U_i   
= C_i - C_iU_i + U_i\Delta,
    \end{align*}} where $\leq_1$ is due to $0 \leq U_i \leq 1$ and $C_i-q_3 \geq 0$.
  \item If $q_3 > C_i$: we have 
{\small    \begin{align*}
work_i(\Delta)  = \;\;& p_3 C_i + C_i \leq p_3 C_i + C_i +  (q_3 -
C_i) U_i = C_i - C_iU_i + U_i\Delta,
    \end{align*}} where $\leq$ is due to $0 \leq U_i \leq 1$ and $q_3-C_i > 0$.
  \end{itemize}
\vspace{-0.2in}
\end{proof}

\begin{lemma}
    \label{lemma:workload-heavy-approximate}  
For any $\Delta \geq 0$,
    \begin{equation}
      \label{eq:work-heavy-upper-linear}
  \omega_i^{heavy}(\Delta) \leq C_i +U_i D_i - C_i U_i+ U_i\Delta.      
    \end{equation}
\end{lemma}
\begin{proof}
  Due to
  Lemma~\ref{lemma:workload-carry-in-tasks}~and~Lemma~\ref{lemma:workload-approximate},
  the inequality holds.  
\end{proof}

\begin{lemma}
    \label{lemma:workload-carry-in-light-tasks-approximate}
    If $U_i \leq \rho \leq 1$, for any $\Delta \geq 0$,
    \begin{equation}
      \label{eq:light-approx}
    \omega^{light}_i(\Delta) \leq C_i - C_iU_i + U_i\Delta.      
    \end{equation}
\end{lemma}
\begin{proof}
  We consider the three upper bounds in
  Lemma~\ref{lemma:workload-carry-in-light-tasks} individually.  
  When $\Delta \leq C_i$, this follows from
  Lemma~\ref{lemma:workload-approximate} directly. When 
  $\Delta > C_i$ and $\omega_i^{light}(\Delta) = work_i(\Delta)$, it
  holds due to Lemma~\ref{lemma:workload-approximate} as well. 
  
  For the last case we have to bound $(p_2+1)C_i + \max\{0, C_i- \rho
  (T_i-q_2)\}$, as defined in
  Lemma~\ref{lemma:workload-carry-in-light-tasks}. By the definition
  of $p_2$ and $q_2$, i.e., $\Delta-C_i = p_2T_i + q_2$, in the statement
  of Lemma~\ref{lemma:workload-carry-in-light-tasks}, we have
  $p_2+1 = \ceiling{(\Delta-C_i)/T_i}$ and $(p_2+1)C_i = work_i(p_2
  T_i + C_i) = work_i(\Delta - q_2)$.  Therefore, for any
  $\Delta > C_i$, if $C_i- \rho (T_i-q_2)  \geq 0$, we get
  \begin{align*}
\omega^{light}_i(\Delta) =\;\;   & (p_2+1)C_i +C_i- \rho (T_i-q_2)   \\
= \;\;& 
work_i(\Delta - q_2) + C_i - \rho (T_i-q_2)    \\
\leq_1\;\; &C_i - C_iU_i +U_i (\Delta - q_2) + C_i - \rho (T_i-q_2)    \\
= \;\;&C_i - C_iU_i +U_i \Delta - q_2 (U_i-\rho) - T_i (\rho - U_i)\\
= \;\;&C_i - C_iU_i +U_i \Delta  + (T_i - q_2) (U_i - \rho)    \\
\leq_2\;\; &C_i - C_iU_i +U_i \Delta,
  \end{align*}
  where $\leq_1$ is due to Lemma~\ref{lemma:workload-approximate} and
  $\leq_2$ is due to $q_2 \leq T_i$ and $U_i \leq \rho$. For any
  $\Delta > C_i$, if $C_i- \rho (T_i-q_2)  < 0$, similarly, we have
 \begin{align*}
   \omega^{light}_i(\Delta) = &(p_2+1)C_i = work_i(p_2 T_i + C_i) \\
  \leq &
   work_i(p_2 T_i + C_i + q_2)\\= &work_i(\Delta) \leq C_i - C_iU_i +U_i \Delta.
  \end{align*}
  Therefore, we reach the conclusion.  
\end{proof}


With
the help of the above lemmas for safe approximations,
we can now
safely and efficiently handle the schedulability test for specific
$\ell$ and $\rho$ in the following theorem. This handles {\bf Issue 1}
explained at the end of Section~\ref{sec:our-tests}.
\begin{theorem}
  \label{thm:schedulability-test-all-range}
  Task $\tau_k$ is schedulable by the given global fixed-priority scheduling
  if
  \begin{align}
& \forall \ell \in \mathbb{N}, \exists 1 \geq \rho \geq
\ell C_k/((\ell-1) T_k + D_k)\nonumber    \\
&\frac{\ell C_k}{D_k'} + \sum_{\tau_i \in {\bf T}^{carry-approx}}
\frac{\gamma_i U_iD_i}{D_k'} +
\sum_{i=1}^{k-1} \left(\frac{C_i - C_iU_i}{D_k'} + U_i\right)\leq \mu_k, \label{eq:quadratic-rho-all-range}
  \end{align}
where $\mu_k= 
M-(M-1)\rho$ with $1 \geq \rho \geq
\ell C_k/((\ell-1) T_k + D_k)$, $D_k'$ is $(\ell-1) T_k+D_k$, 
\begin{equation}\label{eq:gamma-i}
  \gamma_i = 
  \begin{cases}
    1& \mbox{ if } U_i > \rho\\
    0 & \mbox{ if } U_i \leq \rho
  \end{cases}
\end{equation}
and ${\bf T}^{carry-approx}$ is the set of the $\ceiling{\mu_k}-1$ tasks among
the $k-1$  higher-priority tasks with the largest values of $\gamma_i U_iD_i$. 
Note that $|{\bf T}^{carry-approx}|$ can be smaller than $\ceiling{\mu_k}-1$ if 
the number of tasks with $U_i > \rho$ is less than $\ceiling{\mu_k}-1$.
If $D_k \leq T_k$, we only need to consider $\ell=1$. 
\end{theorem}
\begin{proof}
  We prove that the condition in this theorem is a safe upper bound of
  that in Theorem~\ref{thm:schedulability-test-main}. For specific
  $\ell, \rho, \Delta$, we can find ${\bf T}^{carry}$ as defined in
  Theorem~\ref{thm:schedulability-test-main}. 
  By
  Lemmas~\ref{lemma:workload-approximate},~\ref{lemma:workload-heavy-approximate},~and~\ref{lemma:workload-carry-in-light-tasks-approximate}
  and the assumptions $\Delta \geq (\ell-1) T_k+D_k = D_k'$ and $0 < U_i \leq 1 \forall \tau_i$, we have
  \begin{align}
     & \ell C_k + \sum_{\tau_i \in {\bf T}^{carry}} \omega^{diff}_i(\Delta, \rho) + \sum_{i=1}^{k-1}work_i(\Delta)\nonumber\\
    \leq  & \ell C_k + \sum_{\tau_i \in {\bf T}^{carry}} \gamma_i U_i D_i  + \sum_{i=1}^{k-1} \left(C_i - C_i U_i + U_i \Delta\right)\\
    \leq  & \ell C_k + \sum_{\tau_i \in {\bf T}^{carry-approx}} \gamma_i U_i D_i  + \sum_{i=1}^{k-1} \left(C_i - C_i U_i + U_i \Delta\right)\\
    \leq  & \Delta\cdot\left(\frac{\ell C_k}{D_k'} + \sum_{\tau_i \in {\bf T}^{carry-approx}} \frac{\gamma_i U_i D_i}{D_k'}  + \sum_{i=1}^{k-1} \left(\frac{C_i - C_i U_i}{D_k'} + U_i\right)\right)
  \end{align}
  Therefore, the test in Theorem~\ref{thm:schedulability-test-main} can be safely over-approximated as follows:
{\small \begin{align}
& \forall \ell \in \mathbb{N}, \exists 1 \geq \rho \geq
\ell C_k/((\ell-1) T_k + D_k)
T_k+D_k\nonumber \\
& \frac{\ell C_k}{D_k'} +  \left(\sum_{\tau_i \in {\bf T}^{carry-approx}} \frac{\gamma_i U_i D_i}{D_k'}  \right) + \sum_{i=1}^{k-1} \left(\frac{C_i - C_i U_i}{D_k'} + U_i \right)\leq \mu_k
\end{align}}
\end{proof}


Theorem~\ref{thm:schedulability-test-all-range}
provides two interesting implications to handle {\bf Issue
  2}. Firstly, if $U_i \leq \rho$, the linear approximation of
$work_i(\Delta)$ by considering task $\tau_i$ as a non-carry-in task
in Lemma~\ref{lemma:workload-approximate} is the same as the linear
approximation of $\omega_i^{light}(\Delta)$ by considering task $\tau_i$
as a carry-in task in
Lemma~\ref{lemma:workload-carry-in-light-tasks-approximate}. Therefore,
the carry-in tasks are only effective for those tasks $\tau_i$ with
$U_i > \rho$. Secondly, for a specific $\ell$, deciding whether a
specific $\rho$ exists to pass the test in
Eq.~\eqref{eq:quadratic-rho-all-range} can be done by only testing a finite
number of $\rho$ values, i.e.     
by starting from $\rho=\ell C_k/((\ell-1) T_k + D_k$ and increasing
$\rho$ to the next (higher) values where ${\bf T}^{carry-approx}$ changes. This means either  
1)  $\rho = U_i$ for certain higher-priority task $\tau_i$, i.e., the summation
can be larger with the same number of summands; or 2) 
$\mu_k=M-(M-1)\rho$ is an integer, i.e., the number of summands increases.
 This only has 
time complexity $O((k+M)\log (k+M))$, mainly due to the sorting, 
when proper data structures are used. 

\subsection{Linear-Time Schedulability Tests} 
The time complexity of Theorem~\ref{thm:schedulability-test-all-range}
is due to the search of possible $\rho$ values.  Nevertheless, we can
directly set $\rho$ to $U_{\delta,k}^{\max}$ which implies that there
is no carry-in task in the linear-approximation form. With this
simplification, we can conclude different schedulability tests in
Theorems~\ref{thm:schedulability-test-ell=all-range-deltamax},~\ref{thm:schedulability-test-ell=1-or-infty},~and~\ref{thm:linear-time-pessimistic}. Although
these tests are not superior to
Theorem~\ref{thm:schedulability-test-all-range}, our main target is
the test in Theorem~\ref{thm:linear-time-pessimistic}, which will be
used \emph{mainly to derive the speedup bounds later in
Theorem~\ref{thm:global-dm-speedup-factors-3}}.

\begin{theorem}
  \label{thm:schedulability-test-ell=all-range-deltamax}
  Task $\tau_k$ is schedulable by the given global fixed-priority scheduling if
  $\forall \ell \in \mathbb{N}$ 
  \begin{equation}
\frac{\ell C_k}{D_k'} + \sum_{i=1}^{k-1} \left(\frac{C_i - C_iU_i}{D_k'} +
U_i\right)\leq (M - (M-1)U_{\delta,k}^{\max}) \label{eq:quadratic-simplified}
  \end{equation}
holds, where $D_k'$ is $(\ell-1) T_k+D_k$.
\end{theorem}
\begin{proof}
  This comes directly from
  Theorem~\ref{thm:schedulability-test-all-range} by setting $\rho$ to
  $U_{\delta,k}^{\max}$ and the facts that $U_{\delta,k}^{\max} \geq
  U_i$ for $i=1,2,\ldots, k-1$ and $U_{\delta,k}^{\max} \geq \delta_k
  \geq \ell C_k/((\ell-1) T_k + D_k)$ by definition.
\end{proof}

\begin{theorem}
\label{thm:schedulability-test-ell=1-or-infty}
  Suppose that $D_k > T_k$. Let $b$ be $\frac{D_k-T_k}{T_k}$. 
Task $\tau_k$ is schedulable by the given global fixed-priority scheduling
algorithm if:
{\small 
\begin{align}
&\sum_{i=1}^{k} U_i\leq (M - (M-1)U_{\delta,k}^{\max}),& \mbox{ when } b U_k -
\sum_{i=1}^{k-1} \frac{C_i - C_iU_i}{T_k} > 0\\
    &\frac{C_k}{D_k} +
\sum_{i=1}^{k-1} (\frac{C_i - C_iU_i}{D_k} +U_i) \leq (M -
(M-1)U_{\delta,k}^{\max}), &\mbox{ otherwise}
\end{align}}
\end{theorem}
\begin{proof}
  For a given $\ell$, the left-hand side in
  Eq.~\eqref{eq:quadratic-simplified} can be rephrased as:
{\footnotesize  \begin{align}
   F(\ell) = &  \frac{\ell C_k}{D_k'} + \sum_{i=1}^{k-1} \left(\frac{C_i -
      C_iU_i}{D_k'} +
    U_i\right) \nonumber\\
    = &\frac{\ell U_k  +\sum_{i=1}^{k-1} \frac{C_i -
      C_iU_i}{T_k}}{\ell + b} + \sum_{i=1}^{k-1}  U_i
  \end{align}}

  The first order derivative of $F(\ell)$ with respect to $\ell$ is:
  \begin{equation}
   \frac{ \partial F(\ell)}{\partial \ell} =\frac{b U_k - \sum_{i=1}^{k-1} \frac{C_i -
      C_iU_i}{T_k}}{(\ell+b)^2}.
  \end{equation}
  We have to consider two cases:

  \begin{itemize}
  \item Case 1: if $b U_k - \sum_{i=1}^{k-1} \frac{C_i -
      C_iU_i}{T_k} > 0$, then $F(\ell)$ is an increasing function with
    respect to $\ell$. Therefore, $F(\ell)$ is maximized when $\ell
    \rightarrow \infty$, i.e., $F(\ell) \leq \sum_{i=1}^{k} U_i$.
  \item Case 2: if $b U_k - \sum_{i=1}^{k-1} \frac{C_i -
      C_iU_i}{T_k} \leq 0$, then $F(\ell)$ is a non-increasing
    function with respect to $\ell$. Therefore, $F(\ell)$ is maximized
    when $\ell \rightarrow 1$, i.e., $F(\ell) \leq \frac{C_k}{D_k} +
    \sum_{i=1}^{k-1} (\frac{C_i - C_iU_i}{D_k} +U_i)$.
  \end{itemize}
\end{proof}

\begin{theorem}
  \label{thm:linear-time-pessimistic}
  Task $\tau_k$ is schedulable by the given global fixed-priority scheduling if
  \begin{align}
\delta_k+ \sum_{i=1}^{k-1} \left(\frac{C_i - C_iU_i}{D_k} +
U_i\right)\leq M - (M-1)U_{\delta,k}^{\max}
\label{eq:linear-time-pessimistic}
  \end{align}
\end{theorem}
\begin{proof}
  Based on Theorem~\ref{thm:schedulability-test-ell=all-range-deltamax} and
  the two facts that  $D_k'=(\ell-1)T_k+D_k \geq D_k$ and $\delta_k \geq \ell
  C_k/((\ell-1) T_k + D_k)$ for all $\ell \in \mathbb{N}$, we reach the
  conclusion.
\end{proof}


\subsection{Dominance}

We now show analytical dominance among the tests presented above and
in Theorem~\ref{thm:schedulability-test-main} in the following
corollary. A test $\mathcal{B}_1$ analytically dominates another test
$\mathcal{B}_2$ if the schedulability condition in $\mathcal{B}_1$ always
dominates that in $\mathcal{B}_2$. This means, if task $\tau_k$ is deemed
schedulable by $\mathcal{B}_2$, task $\tau_k$ is also deemed schedulable
by $\mathcal{B}_1$.
\begin{corollary}
  \label{corollary:dominance}
  For arbitrary-deadline sporadic real-time systems under global fixed-priority
  scheduling, the schedulability tests have the following dominance
  relations.
  \begin{compactitem}
  \item  Theorem~\ref{thm:schedulability-test-main} analytically dominates
    Theorem~\ref{thm:schedulability-test-all-range}.
  \item  Theorem~\ref{thm:schedulability-test-all-range} analytically dominates
    Theorem~\ref{thm:schedulability-test-ell=all-range-deltamax}.
  \item Theorem~\ref{thm:schedulability-test-ell=all-range-deltamax} is equivalent to 
    the test in Theorem~\ref{thm:schedulability-test-ell=1-or-infty}.
  \item Theorem~\ref{thm:schedulability-test-ell=1-or-infty} analytically dominates
    Theorem~\ref{thm:linear-time-pessimistic}.
    \end{compactitem}
\end{corollary}
\begin{proof}
  They follow directly from the above analyses. The reason why
  Theorems~\ref{thm:schedulability-test-ell=all-range-deltamax}~and~\ref{thm:schedulability-test-ell=1-or-infty}
  are equivalent is because the conditions in
  Theorem~\ref{thm:schedulability-test-ell=1-or-infty} represent
  exactly the worst-case $\ell$ selection in
  Theorem~\ref{thm:schedulability-test-ell=1-or-infty}. The other
  cases are obvious.
\end{proof}

Although we will show in
Theorem~\ref{thm:global-dm-speedup-factors-3-tests} that all the above
schedulability tests have the same speedup bound for global DM, the
\emph{performance} of the schedulability tests in this section
can be very different \emph{in practice}.  Chen et al.~\cite{DBLP:conf/ecrts/ChenBHD17}
have recently shown that ``\emph{Speedup factors ... often lack the
  power to discriminate between the performance of different
  scheduling algorithms and schedulability tests even though the
  performance of these algorithms and tests may be very different when
  viewed from the perspective of empirical evaluation}.'' To avoid
concluding an algorithm with a reasonable speedup bound but practically not
useful, we performed a series of experiments and present the results
in Section~\ref{sec:eval}.

\section{Global Deadline-Monotonic (DM) Scheduling}
\label{sec:global-DM}

After presenting the schedulability tests for any global
fixed-priority scheduling algorithms, we focus ourselves on global DM
in this section. We will discuss the speedup upper bound and 
the speedup lower bound.
Baruah and Fisher \cite{DBLP:conf/opodis/BaruahF07} showed that global
DM has a speedup upper bound of $2~+~\sqrt{3}~\approx~3.73$ 
compared to the
optimal schedules, based on the test restated in
Theorem~\ref{thm:baruah-fisher07-bug-free}.  This is the best known upper bound
on speedup factors  for arbitrary-deadline sporadic task systems under
global fixed-priority scheduling. 
Evaluating $\load(k)$ in Theorem~\ref{thm:baruah-fisher07-bug-free}
requires to calculate $\sum_{i=1}^{k} \dbf(\tau_i, t)/t$ at all
time points $t$. This means, the na\"{i}ve implementation has an 
exponential-time complexity.  There are more efficient methods, as
discussed by Baruah and Bini \cite{BaruahBini-DASIP08}, but the time
complexity remains exponential. Although it is possible to approximate
$\load(k)$ by using approximate demand bound functions in polynomial
time, this is at a price of higher $\load(k)$. We show that
the test in Theorem~\ref{thm:baruah-fisher07-bug-free} is
over-pessimistic and is analytically dominated by our linear-time
schedulability test in Theorem~\ref{thm:linear-time-pessimistic} under
global DM.

\begin{corollary}
  \label{cor:superior-to-BaruahFisher2007}
  For global DM, the schedulability test in
  Theorem~\ref{thm:linear-time-pessimistic} analytically dominates the
  schedulability test in Theorem~\ref{thm:baruah-fisher07-bug-free}
  proposed by Baruah and Fisher \cite{DBLP:conf/opodis/BaruahF07}.
\end{corollary}
\begin{proof}
  This is due to the following facts:
  \begin{itemize}
  \item By definition, $\load(k) \geq \mbox{limit}_{t\rightarrow \infty}\sum_{i=1}^{k} \dbf(\tau_i,
    t)/t = \sum_{i=1}^{k} U_i$.
  \item Since $D_i \leq D_k$ in global DM for $i=1,2,\ldots,k-1$, we
    know that $\frac{\sum_{i=1}^{k} \dbf(\tau_i,D_k)}{D_k} \geq
    \sum_{i=1}^{k} \frac{C_i}{D_k}$. Therefore, $\load(k) \geq \sum_{i=1}^{k}
    \frac{C_i}{D_k}$.
  \end{itemize}
  Combining these 
  facts, we get 
  \begin{align}
    \delta_k+
  \sum_{i=1}^{k-1} \left(\frac{C_i - C_iU_i}{D_k} + U_i\right) 
\leq  \sum_{i=1}^{k} \frac{C_i}{D_k}  + U_i \leq 2\load(k).
  \end{align}
  Since we know that the right-hand side  in
  Eq.~\eqref{eq:sch-test-corollary-v1}, i.e., \mbox{$M -
  (M-1)\delta_{\max}(k)$}, is less than or equal to \mbox{$M -
  (M-1)U_{\delta,k}^{\max}$} in Eq.~\eqref{eq:linear-time-pessimistic}, we reach
  the conclusion.
\end{proof}

\begin{theorem}
\label{thm:global-dm-speedup-factors-3}
  Global DM has a speedup bound of $3-\frac{1}{M}$, with respect to
  the optimal schedule, when $M \geq 2$.
\end{theorem}
\begin{proof}
  We only prove the speedup bound by using the schedulability test in
  Theorem \ref{thm:linear-time-pessimistic}. Due to the dominance
  properties in Corollary~\ref{corollary:dominance}, such a bound also
  holds for the schedulability tests from 
  Theorems~\ref{thm:schedulability-test-main},~\ref{thm:schedulability-test-all-range},~\ref{thm:schedulability-test-ell=all-range-deltamax},~and~\ref{thm:schedulability-test-ell=1-or-infty}.

  Suppose that task $\tau_k$ is not schedulable by global DM.  Since
  $D_i \leq D_k$ for any $i=1,2,\ldots,k-1$ under global DM, we know
  $\dbf(\tau_i,D_k) \geq C_i$. Therefore, under global DM,
  $\sum_{i=1}^{k} \frac{C_i}{M D_k} \leq \sum_{i=1}^{k}
  \frac{\dbf(\tau_i,D_k)}{M D_k} \leq \frac{\sum_{\tau_i \in {\bf T}}
    \dbf(\tau_i, D_k)}{M D_k} \leq \max_{t>0} \frac{\sum_{\tau_i \in
      {\bf T}} \dbf(\tau_i, t)}{Mt}$.  

  By the assumption that task $\tau_k$ is
  also deemed not schedulable by
  Theorem~\ref{thm:linear-time-pessimistic}, we have
  \begin{align}
&    \delta_k+ \sum_{i=1}^{k-1} \left(\frac{C_i - C_iU_i}{D_k} +
      U_i\right)>  M - (M-1)U_{\delta,k}^{\max}     \nonumber\\
\Rightarrow&  \sum_{i=1}^{k} \frac{C_i}{M D_k} +
     \sum_{i=1}^{k} \frac{U_i}{M} > 1 - \left(1-\frac{1}{M}\right)
     U_{\delta,k}^{\max}\nonumber\\
\Rightarrow&  \sum_{i=1}^{k} \frac{C_i}{M D_k} +
     \sum_{i=1}^{k} \frac{U_i}{M} +  \left(1-\frac{1}{M}\right)
     U_{\delta,k}^{\max}> 1
  \end{align}

  Therefore, either $\max_{t>0} \frac{\sum_{\tau_i \in {\bf T}}
    \dbf(\tau_i, t)}{Mt} \geq \sum_{i=1}^{k} \frac{C_i}{M D_k} >
  \frac{1}{3-1/M}$, or $\sum_{i=1}^{k} \frac{U_i}{M} >
  \frac{1}{3-1/M}$, or $\delta_{\max}(k) \geq U_{\delta,k}^{\max} >
  \frac{1}{3-1/M}$. By Lemma~\ref{lemma:lower-speed-bound}, we reach
  the conclusion of the speedup bound for global DM with respect to
  the optimal schedule.
\end{proof}

\begin{theorem}
  \label{thm:global-dm-speedup-factors-3-tests}
  For global DM, the schedulability tests in
  Theorems~\ref{thm:schedulability-test-main},~\ref{thm:schedulability-test-all-range},~\ref{thm:schedulability-test-ell=all-range-deltamax},~\ref{thm:schedulability-test-ell=1-or-infty},~and~\ref{thm:linear-time-pessimistic}
  have a speedup bound of $3-\frac{1}{M}$, with respect to the
  optimal schedule, when $M \geq 2$.
\end{theorem}
\begin{proof}
  This is due to 
  Theorem~\ref{thm:global-dm-speedup-factors-3} and
  Corollary~\ref{corollary:dominance}, because all of the tests in
  Theorems~\ref{thm:schedulability-test-main},~\ref{thm:schedulability-test-all-range},~\ref{thm:schedulability-test-ell=all-range-deltamax},~\ref{thm:schedulability-test-ell=1-or-infty}
  dominate the test in Theorem~\ref{thm:linear-time-pessimistic} as
  presented in Corollary~\ref{corollary:dominance}.
\end{proof}

\begin{theorem}
\label{theorem:sporadic-arbitrary-tight}
The speedup bound of global DM for arbitrary-deadline task systems is
at least $3-\frac{3}{M+1}$.
\end{theorem}
\begin{proof}
  The proof is based on a concrete task set.  We specifically use
  the following task set ${\bf T}^{ad}$ with $N=2M+1$ tasks.  Let
  $\varepsilon$ be an arbitrarily small positive real number such that
  $1/\varepsilon$ is an integer.  Let $\eta \ll \varepsilon$ be an
  arbitrarily small positive number, that is 
  used to enforce the priority assignment under global DM:
\begin{itemize}
  \item $C_i=\frac{\varepsilon}{3}$, $T_i=\varepsilon$, $D_i=1$, for $i=1,2,\ldots, M$.
  \item $C_i=\frac{1}{3}$, $T_i=\infty$, $D_i=1+\eta$, for $i=M+1,M+2,\ldots, 2M$.
  \item $C_i=\frac{1+\varepsilon}{3}$, $T_i=\infty$, $D_i=1+2\eta$, for $i=2M+1$
  \end{itemize}
  As the setting of $\eta \ll \varepsilon$ is just to enforce the
  indexing, \emph{we will directly take $\eta \rightarrow 0$ here.}
  In the Appendix, we prove two properties: 1) ${\bf T}^{ad}$ is not
  schedulable by global DM under a concrete instance which releases
  all the tasks at time $0$ and the subsequent jobs periodically. 2)
  There exists a feasible schedule for task set ${\bf T}^{ad}$ at any
  speed no lower than $\frac{1+\varepsilon}{3} +
  \frac{1+\varepsilon}{3M}$ under a concrete semi-partitioned
  multiprocessor schedule, i.e., $\setof{\tau_{m}, \tau_{m+M}}$
  assigned to processor $m$ for $m=1,2,\ldots, M$ and task
  $\tau_{2M+1}$ executed partially on each of the $M$ processors.
  Therefore, a lower bound on the speedup bound of global DM is:
\[
\lim_{\varepsilon \rightarrow 0}\frac{1}{\frac{1+\varepsilon}{3} + \frac{1+\varepsilon}{3M}} = 
\lim_{\varepsilon \rightarrow 0}\frac{3M}{(1+\varepsilon)\times(M+1)} = \frac{3M}{M+1} = 3- \frac{3}{M+1}.
\]
\end{proof}

By
Theorems~\ref{thm:global-dm-speedup-factors-3}~and~\ref{theorem:sporadic-arbitrary-tight},
we can reach the conclusion that all the schedulability tests from
Theorems~\ref{thm:schedulability-test-main},~\ref{thm:schedulability-test-all-range},~\ref{thm:schedulability-test-ell=all-range-deltamax},~\ref{thm:schedulability-test-ell=1-or-infty},~and~\ref{thm:linear-time-pessimistic}
are asymptotically tight with respect to speedup bounds. However, due to dominance properties in
Corollary~\ref{corollary:dominance}, these tests clearly have
different performance with respect to schedulability tests.

\section{Evaluation}
\label{sec:eval}

We evaluated 
the scheduling
tests provided in this paper by comparing their acceptance ratio to the
acceptance ratio of other algorithms, i.e., comparing the percentage of task
sets accepted for the different schedulability tests,
%
using different settings for the number of processors, the
scheduling policy, and the
ratio of the relative deadline
to the period.

\textbf{Evaluation Setup:}
We conducted evaluations for
homogeneous multiprocessor systems with $M=4$, $M=8$, and $M=16$ processors. We generated
$100$ task sets with cardinality of both $N=5\times M$ and $N=10 \times M$, and
utilization ranging from $M\times 5\%$ to $M \times 100\%$ in steps of
$M\times 5\%$. 
The UUniFast-Discard method~\cite{DBLP:journals/rts/BiniB05} was adopted to
generate the utilization values of a set of $N$ tasks under the target
utilization.
As suggested by Emberson et
al.~\cite{emberson2010techniques}, the periods were generated according to a
log-uniform distribution, 
 with 1, 2, and 3 orders of  
magnitude, i.e., $[1ms-10ms]$, $[1ms-100ms]$, and $[1ms-1000ms]$. 
For each task, the relative deadline was set to
the period multiplied with a  value randomly drawn under a uniform distribution
from a given interval $I$. We conducted evaluations using different interval,
i.e., 
$I$ was $[0.8,2],[0.8,5],[0.8,10],[1,2],[1,5],$ or $[1,10]$.
To schedule the task sets, we applied global deadline-monotonic (DM) and global
slack-monotonic (SM)~\cite{DBLP:conf/opodis/Andersson08} scheduling.

Whether the task set is schedulable under the given scheduling
approach or not was tested
using the following schedulability tests: 
\begin{compactitem}
  \item LOAD: The load-based analysis by Baruah and Fisher
  in~\cite{DBLP:conf/icdcn/BaruahF08}, only for DM scheduling.
  \item BAK: The test  by Baker in Theorem
  11 in~\cite{baker2006analysis}. 
  \item HC: The sufficient test in Corollary 2 by Huang and Chen
  in~\cite{DBLP:conf/rtns/HuangC15}.
  \item OUR-4.4: The sufficient test in
  Theorem~\ref{thm:schedulability-test-all-range} in this paper.
  \item OUR-4.6: The sufficient test in
  Theorem~\ref{thm:schedulability-test-ell=1-or-infty} in this paper.
  \item OUR-4.7: The sufficient test in
  Theorem~\ref{thm:linear-time-pessimistic} in this paper.
\end{compactitem}
We also  checked if a task set was schedulable according to at least
one of the tests, 
denoted as \emph{ALL}. 
\emph{We only present a small set of the conducted tests here. The
diagrams  of all conducted evaluations  can be found in~\cite{eval_zip}.}

\textbf{Evaluation Results:} 
Figure~\ref{fig:kevins_eval} shows the evaluations under the setting used in the
paper by Huang and Chen~\cite{DBLP:conf/rtns/HuangC15}. They used DM scheduling
on $M=8$ processors, a task set containing 40 tasks and ratios of
$\frac{D_i}{T_i}\in[0.8,2]$ and analyzed the schedulability for 
$T_i$ values that differ up to 1, 2, and 3 orders of magnitude, i.e., $T_i$ in a
range of $[1ms,10ms]$, $[1ms,100ms]$, or $[1ms,1000ms]$. 
The test by Baruah and  Fisher~\cite{DBLP:conf/icdcn/BaruahF08} is clearly
outperformed by 
Theorem~\ref{thm:schedulability-test-ell=1-or-infty}, 
Theorem~\ref{thm:linear-time-pessimistic}, and
Baker's test~\cite{baker2006analysis}, which provide similar acceptance ratios. 
The test by Huang and Chen~\cite{DBLP:conf/rtns/HuangC15} outperforms those three
tests and is  worse than the test in
Theorem~\ref{thm:schedulability-test-all-range} in these settings. However,
there is no dominance relation between Theorem~\ref{thm:schedulability-test-all-range} and the
test by Huang and Chen~\cite{DBLP:conf/rtns/HuangC15}, as  some task sets
are schedulable under the
test by Huang and Chen~\cite{DBLP:conf/rtns/HuangC15} but not schedulable under
Theorem~\ref{thm:schedulability-test-all-range} and vise versa.


\begin{figure*}[t]
   \centering
  \includegraphics[trim = 0mm 80mm 0mm 0mm, clip,
width=0.9\textwidth]{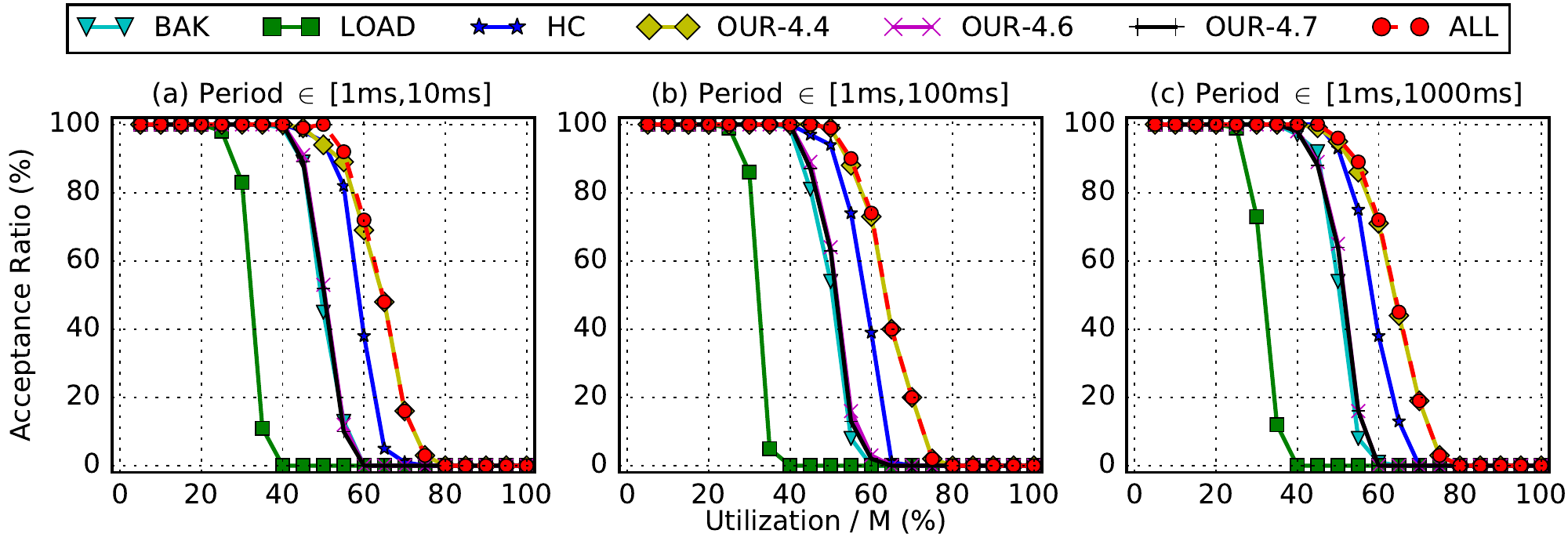}
\vspace{-0.15in}   \caption{\small Comparison of the tests presented in
Theorem~\ref{thm:schedulability-test-all-range},
\ref{thm:schedulability-test-ell=1-or-infty}, and
\ref{thm:linear-time-pessimistic} with the methods from
Baruah and  Fisher~(LOAD)~\cite{DBLP:conf/icdcn/BaruahF08},
Baker~\cite{baker2006analysis}, and Huang and
Chen~\cite{DBLP:conf/rtns/HuangC15} for different ranges of period. The
evaluation setup is the same as in~\cite{DBLP:conf/rtns/HuangC15}, i.e., DM, $M=8$, $N=40$,
$\frac{D_i}{T_i}\in[0.8,2]$.}
    \label{fig:kevins_eval}
\end{figure*}

\begin{figure*}[t]
   \centering
   \vspace{-0.1in}
  \includegraphics[trim = 0mm 80mm 0mm 0mm, clip,
width=0.9\textwidth]{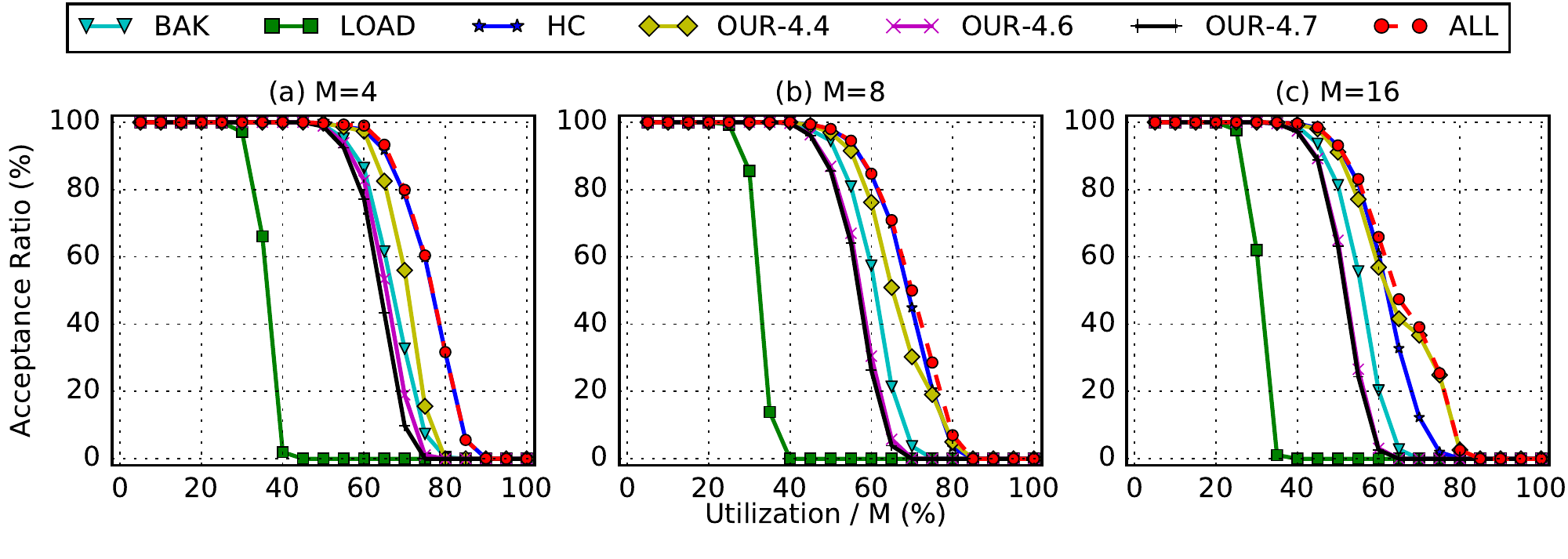}
\vspace{-0.15in}   \caption{\small Comparison of the tests presented in
Theorem~\ref{thm:schedulability-test-all-range},
\ref{thm:schedulability-test-ell=1-or-infty}, and
\ref{thm:linear-time-pessimistic} with the methods from
Baruah and  Fisher~(LOAD)~\cite{DBLP:conf/icdcn/BaruahF08},
Baker~\cite{baker2006analysis}, and Huang and
Chen~\cite{DBLP:conf/rtns/HuangC15} for different $M$ values. The other
parameters are fixed, i.e., DM, $N=5\times M$,
$T_i\in[1ms,10ms]$, and $\frac{D_i}{T_i}\in[0.8,10]$.}
    \label{fig:proc_eval}
\vspace{-0.15in}
\end{figure*}

There are other configurations where the test by Huang and
Chen~\cite{DBLP:conf/rtns/HuangC15} performs better than
Theorem~\ref{thm:schedulability-test-all-range}. One example is shown in
Figure~\ref{fig:proc_eval},  analyzing the impact of the number of
processors. Here Theorem~\ref{thm:schedulability-test-all-range} performs
compatible to Theorem~\ref{thm:schedulability-test-ell=1-or-infty}, 
Theorem~\ref{thm:linear-time-pessimistic}, and
Baker's test~\cite{baker2006analysis} for $M=4$. 
When the number of processors increases, Theorem~\ref{thm:schedulability-test-all-range}
performs better. The gap to Huang and
Chen~\cite{DBLP:conf/rtns/HuangC15} is smaller for 8 processors 
and Theorem~\ref{thm:schedulability-test-all-range} 
has a higher acceptance rate
when the utilization level is $80\% \times M$. For $M=16$ processors
Theorem~\ref{thm:schedulability-test-all-range} accepts more task sets than
Huang and Chen~\cite{DBLP:conf/rtns/HuangC15} when the utilization level is $\geq 65\%\times M$. In
addition, it is possible that the number of task sets that is accepted by
at least one algorithm is not close to the number of task sets accepted by Huang
and Chen~\cite{DBLP:conf/rtns/HuangC15} or
Theorem~\ref{thm:schedulability-test-all-range} as can be seen for 
the utilization level $75\%\times M$ in the case where $M=8$. 

Furthermore, we  tracked if  the test by
Baker~\cite{baker2006analysis} accepted some task sets that were not accepted by Huang
and Chen~\cite{DBLP:conf/rtns/HuangC15} or
Theorem~\ref{thm:schedulability-test-all-range}, which happened occasionally. 
Therefore, we conclude that there is no dominance relation between any of
those three  tests, i.e., Theorem~\ref{thm:schedulability-test-all-range},
and the tests by Baker~\cite{baker2006analysis} and 
by 
Huang and
Chen~\cite{DBLP:conf/rtns/HuangC15}.
 As these 
 tests 
  can all be implemented with polynomial-time complexity, all three  
  should be applied when testing the schedulability of  arbitrary-deadline
  task sets under global fixed-priority scheduling.


\section{Conclusion}

We present a series of schedulability tests for multiprocessor
systems under any given fixed-priority scheduling approach. Those schedulability
tests have different tradeoffs between their accuracy and their time complexity.
All those schedulability tests dominate the approach by Baruah and
Fisher~\cite{DBLP:conf/icdcn/BaruahF08}, both with respect to speedup bounds and
schedulability analysis. Theorem~\ref{thm:schedulability-test-main} is the most powerful
schedulability test in this paper. However, we do not reach any
concrete implementation with affordable time complexity. In the future
work, we will seek for efficient methods to implement the
schedulability test in Theorem~\ref{thm:schedulability-test-main}.


\bibliographystyle{abbrv}
\bibliography{real-time}

\begin{thebibliography}{10}

\bibitem{DBLP:conf/opodis/Andersson08}
B.~Andersson.
\newblock Global static-priority preemptive multiprocessor scheduling with
  utilization bound 38{\%}.
\newblock In {\em Principles of Distributed Systems, 12th International
  Conference, {OPODIS}}, pages 73--88, 2008.

\bibitem{DBLP:conf/rtss/AnderssonBJ01}
B.~Andersson, S.~K. Baruah, and J.~Jonsson.
\newblock Static-priority scheduling on multiprocessors.
\newblock In {\em Real-Time Systems Symposium {(RTSS)}}, pages 193--202, 2001.

\bibitem{baker2006analysis}
T.~P. Baker.
\newblock An analysis of fixed-priority schedulability on a multiprocessor.
\newblock {\em Real-Time Systems}, 32(1-2):49--71, 2006.

\bibitem{DBLP:conf/opodis/BakerC07}
T.~P. Baker and M.~Cirinei.
\newblock Brute-force determination of multiprocessor schedulability for sets
  of sporadic hard-deadline tasks.
\newblock In {\em Principles of Distributed Systems, 11th International
  Conference, {OPODIS}}, pages 62--75, 2007.

\bibitem{baruah2007techniques}
S.~Baruah.
\newblock Techniques for multiprocessor global schedulability analysis.
\newblock In {\em Proceedings of the 28th IEEE International Real-Time Systems
  Symposium}, pages 119--128, 2007.

\bibitem{BaruahBini-DASIP08}
S.~Baruah and E.~Bini.
\newblock Partitioned scheduling of sporadic task systems: an {ILP}-based
  approach.
\newblock In {\em Proc. DASIP}, 2008.

\bibitem{DBLP:journals/rts/BaruahBMS10}
S.~K. Baruah, V.~Bonifaci, A.~Marchetti{-}Spaccamela, and S.~Stiller.
\newblock Improved multiprocessor global schedulability analysis.
\newblock {\em Real-Time Systems}, 46(1):3--24, 2010.

\bibitem{DBLP:conf/opodis/BaruahF07}
S.~K. Baruah and N.~Fisher.
\newblock Global deadline-monotonic scheduling of arbitrary-deadline sporadic
  task systems.
\newblock In {\em Principles of Distributed Systems, 11th International
  Conference, {OPODIS} 2007, Guadeloupe, French West Indies, December 17-20,
  2007. Proceedings}, pages 204--216, 2007.

\bibitem{DBLP:conf/icdcn/BaruahF08}
S.~K. Baruah and N.~Fisher.
\newblock Global fixed-priority scheduling of arbitrary-deadline sporadic task
  systems.
\newblock In {\em Distributed Computing and Networking, 9th International
  Conference, {ICDCN}}, pages 215--226, 2008.

\bibitem{Baruah90}
S.~K. Baruah, A.~K. Mok, and L.~E. Rosier.
\newblock Preemptively scheduling hard-real-time sporadic tasks on one
  processor.
\newblock In {\em In Proceedings of the 11th Real-Time Systems Symposium},
  pages 182--190, 1990.

\bibitem{DBLP:conf/opodis/BertognaCL05}
M.~Bertogna, M.~Cirinei, and G.~Lipari.
\newblock New schedulability tests for real-time task sets scheduled by
  deadline monotonic on multiprocessors.
\newblock In {\em Principles of Distributed Systems, 9th International
  Conference, {OPODIS}}, pages 306--321, 2005.

\bibitem{DBLP:journals/tc/BiniB04}
E.~Bini and G.~C. Buttazzo.
\newblock Schedulability analysis of periodic fixed priority systems.
\newblock {\em {IEEE} Trans. Computers}, 53(11):1462--1473, 2004.

\bibitem{DBLP:journals/rts/BiniB05}
E.~Bini and G.~C. Buttazzo.
\newblock Measuring the performance of schedulability tests.
\newblock {\em Real-Time Systems}, 30(1-2):129--154, 2005.

\bibitem{DBLP:journals/tc/BiniNRB09}
E.~Bini, T.~H.~C. Nguyen, P.~Richard, and S.~K. Baruah.
\newblock A response-time bound in fixed-priority scheduling with arbitrary
  deadlines.
\newblock {\em IEEE Trans. Computers}, 58(2):279--286, 2009.

\bibitem{bini-RTSS2015}
E.~Bini, A.~Parri, and G.~Dossena.
\newblock A quadratic-time response time upper bound with a tightness property.
\newblock In {\em {IEEE} Real-Time Systems Symposium (RTSS)}, 2015.

\bibitem{DBLP:conf/soda/BonifaciCMM10}
V.~Bonifaci, H.-L. Chan, A.~Marchetti-Spaccamela, and N.~Megow.
\newblock Algorithms and complexity for periodic real-time scheduling.
\newblock In {\em SODA}, pages 1350--1359, 2010.

\bibitem{report-chen-baruah-errata-arbitrary-deadline-2007}
J.-J. Chen.
\newblock Erratum: Global deadline-monotonic scheduling of arbitrary-deadline
  sporadic task systems, 2017.
\newblock
  \url{http://ls12-www.cs.tu-dortmund.de/daes/media/documents/publications/downloads/2016-chen-erratum-globalDM.pdf}.

\bibitem{DBLP:journals/corr/abs-k2q}
J.-J. Chen, W.-H. Huang, and C.~Liu.
\newblock k2q: {A} quadratic-form response time and schedulability analysis
  framework for utilization-based analysis.
\newblock In {\em Real-Time Systems Symposium, {RTSS}}, pages 351--362, 2016.

\bibitem{DBLP:conf/ecrts/ChenBHD17}
J.-J. Chen, G.~von~der Br{\"{u}}ggen, W.-H. Huang, and R.~I. Davis.
\newblock On the pitfalls of resource augmentation factors and utilization
  bounds in real-time scheduling.
\newblock In {\em Euromicro Conference on Real-Time Systems, {ECRTS}}, pages
  9:1--9:25, 2017.

\bibitem{eval_zip}
J.-J. Chen, G.~von~der Br\"uggen, and N.~Ueter.
\newblock Evaluation results: Push forward: Global fixed-priority scheduling of
  arbitrary-deadline sporadic task systems.
\newblock
  \url{http://ls12-www.cs.tu-dortmund.de/daes/media/documents/publications/downloads/eval_push_forward.zip}.

\bibitem{DavisRTSS2008}
R.~I. Davis and A.~Burns.
\newblock Response time upper bounds for fixed priority real-time systems.
\newblock In {\em Real-Time Systems Symposium, 2008}, pages 407--418, Nov 2008.

\bibitem{DBLP:journals/csur/DavisB11}
R.~I. Davis and A.~Burns.
\newblock A survey of hard real-time scheduling for multiprocessor systems.
\newblock {\em {ACM} Comput. Surv.}, 43(4):35, 2011.

\bibitem{EisenbrandR08}
F.~Eisenbrand and T.~Rothvo{\ss}.
\newblock Static-priority real-time scheduling: Response time computation is
  np-hard.
\newblock In {\em Proceedings of the 29th {IEEE} Real-Time Systems Symposium,
  {RTSS}}, pages 397--406, 2008.

\bibitem{DBLP:conf/soda/EisenbrandR10}
F.~Eisenbrand and T.~Rothvo{\ss}.
\newblock {EDF}-schedulability of synchronous periodic task systems is
  co{NP}-hard.
\newblock In {\em SODA}, pages 1029--1034, 2010.

\bibitem{DBLP:conf/rtss/Ekberg015}
P.~Ekberg and W.~Yi.
\newblock Uniprocessor feasibility of sporadic tasks remains conp-complete
  under bounded utilization.
\newblock In {\em 2015 {IEEE} Real-Time Systems Symposium, {RTSS}}, pages
  87--95, 2015.

\bibitem{DBLP:conf/ecrts/Ekberg015}
P.~Ekberg and W.~Yi.
\newblock Uniprocessor feasibility of sporadic tasks with constrained deadlines
  is strongly co{NP}-{C}omplete.
\newblock In {\em 27th Euromicro Conference on Real-Time Systems, {ECRTS}},
  pages 281--286, 2015.

\bibitem{emberson2010techniques}
P.~Emberson, R.~Stafford, and R.~I. Davis.
\newblock Techniques for the synthesis of multiprocessor tasksets.
\newblock In {\em International Workshop on Analysis Tools and Methodologies
  for Embedded and Real-time Systems (WATERS 2010)}, pages 6--11, 2010.

\bibitem{geeraerts2013multiprocessor}
G.~Geeraerts, J.~Goossens, and M.~Lindstr{\"o}m.
\newblock Multiprocessor schedulability of arbitrary-deadline sporadic tasks:
  complexity and antichain algorithm.
\newblock {\em Real-time systems}, 49(2):171--218, 2013.

\bibitem{DBLP:conf/rtss/GuanSYY09}
N.~Guan, M.~Stigge, W.~Yi, and G.~Yu.
\newblock New response time bounds for fixed priority multiprocessor
  scheduling.
\newblock In {\em {IEEE} Real-Time Systems Symposium}, pages 387--397, 2009.

\bibitem{DBLP:conf/rtns/HuangC15}
W.-H. Huang and J.-J. Chen.
\newblock Response time bounds for sporadic arbitrary-deadline tasks under
  global fixed-priority scheduling on multiprocessors.
\newblock In {\em RTNS}, 2015.

\bibitem{10.2307/2318871}
V.~Klee.
\newblock Can the measure of $\cup_{i=1}^{n} [a_i, b_i]$ be computed in less
  than {$O(n log n)$} steps?
\newblock {\em The American Mathematical Monthly}, 84(4):284--285, 1977.

\bibitem{DBLP:conf/rtss/Lehoczky90}
J.~P. Lehoczky.
\newblock Fixed priority scheduling of periodic task sets with arbitrary
  deadlines.
\newblock In {\em RTSS}, pages 201--209, 1990.

\bibitem{liu73scheduling}
C.~L. Liu and J.~W. Layland.
\newblock Scheduling algorithms for multiprogramming in a hard-real-time
  environment.
\newblock {\em Journal of the ACM}, 20(1):46--61, 1973.

\bibitem{DBLP:conf/rtas/Lundberg02}
L.~Lundberg.
\newblock Analyzing fixed-priority global multiprocessor scheduling.
\newblock In {\em Real-Time and Embedded Technology and Applications Symposium
  {(RTAS})}, pages 145--153, 2002.

\bibitem{sjodin1998improved}
M.~Sjodin and H.~Hansson.
\newblock Improved response-time analysis calculations.
\newblock In {\em Real-Time Systems Symposium, 1998. Proceedings., The 19th
  IEEE}, pages 399--408. IEEE, 1998.

\bibitem{DBLP:journals/rts/SunL16}
Y.~Sun and G.~Lipari.
\newblock A pre-order relation for exact schedulability test of sporadic tasks
  on multiprocessor global fixed-priority scheduling.
\newblock {\em Real-Time Systems}, 52(3):323--355, 2016.

\bibitem{DBLP:conf/rtcsa/SunLA014}
Y.~Sun, G.~Lipari, N.~Guan, and W.~Yi.
\newblock Improving the response time analysis of global fixed-priority
  multiprocessor scheduling.
\newblock In {\em International Conference on Embedded and Real-Time Computing
  Systems and Applications}, pages 1--9, 2014.

\end{thebibliography}

\section{Appendix: Additional Proofs}
\label{sec:proof}

\subsection{Proof of Theorem~\ref{theorem:sporadic-arbitrary-tight}: Speedup lower bound of global DM
for arbitrary-deadline task systems} 
\label{sec:lower-bound-arbitrary-DM}

 We will specifically use the
following task set ${\bf T}^{ad}$ with $N=2M+1$ tasks.  Let $\varepsilon$
be an arbitrarily small positive real number such that $1/\varepsilon$ is
an integer.  Let $\eta \ll \varepsilon$ be an arbitrarily small positive
number, that is  
used  to enforce the priority assignment under global DM:
\begin{itemize}
  \item $C_i=\frac{\varepsilon}{3}$, $T_i=\varepsilon$, $D_i=1$, for $i=1,2,\ldots, M$.
  \item $C_i=\frac{1}{3}$, $T_i=\infty$, $D_i=1+\eta$, for $i=M+1,M+2,\ldots, 2M$.
  \item $C_i=\frac{1+\varepsilon}{3}$, $T_i=\infty$, $D_i=1+2\eta$, for $i=2M+1$
  \end{itemize}
 As the setting of $\eta \ll \varepsilon$ is just to enforce the indexing, 
 \emph{we will directly take $\eta \rightarrow 0$ 
  here.}

\begin{lemma}
  \label{lemma:lb-arbitrary-deadline-DM-fail}
  ${\bf T}^{ad}$ is not schedulable by global DM.
\end{lemma}
\begin{proof}
  This can be proved by 
  showing that  task $\tau_N$ misses its deadline
  in the following concrete arrival pattern:
  all tasks
  release their first jobs at time $0$ and the subsequent jobs arrive
  as early as possible while respecting their minimum inter-arrival
  times.
  For this arrival pattern, the jobs of tasks $\tau_1, \tau_2, \ldots,
  \tau_M$ are executed from time $i \varepsilon$ to time $i \varepsilon +
  \frac{\varepsilon}{3}$ for $i=0,1,2,\ldots, 1/\varepsilon$. Therefore, 
   these $M$ tasks are executed for in total $1/3$ time units  from time
  $0$ to time $1$. 
  For tasks $\tau_{1+M}, \tau_{2+M}, \ldots,   \tau_{2M}$,  each
  of them is
  executed for $1/3$ time units from time $0$ to time $1$ when the
  processors do not execute 
  $\tau_1, \tau_2, \ldots,
  \tau_M$. 
  Task $\tau_{2M+1}$ is executed alone
  without any overlap with the executions of the higher-priority
  tasks. Therefore task $\tau_{2M+1}$ misses its deadline since it
  needs $\frac{1+\varepsilon}{3}$ time units, but only $\frac{1}{3}$ time
  units are available before its deadline.
\end{proof}

\begin{lemma}
  \label{lemma:lb-arbitrary-deadline-partitioned-succeed}
  There exists a feasible schedule for task set ${\bf T}^{ad}$ at
  any speed no lower than $\frac{1+\varepsilon}{3} + \frac{1+\varepsilon}{3M}$.
\end{lemma}
\begin{proof}
  We will apply multiprocessor semi-partitioned scheduling, in which
  tasks in $\setof{\tau_{m}, \tau_{m+M}}$ are assigned to processor $m$
  for $m=1,2,\ldots, M$.
  In our designed semi-partitioned schedule, a job of task
  $\tau_{2M+1}$, i.e., a part of $\tau_N$, is executed partially on each of the
  $M$ processors as follows: it runs on processor $m$ for $C_{N}/M$
  amount of time, and then migrates to processor $m+1$ to continue its
  execution, for $m=1,2,\ldots,M-1$. To ensure that the
  migration can be served immediately, 
  $\tau_{N}$ is given the 
  the highest-priority  in this schedule.
   Therefore, a
  \emph{subtask} of task $\tau_{N}$ on
  processor $m$, denoted as $\tau_{N,m}$,  has a relative deadline $C_{N}/M$.
  As long as the speed of the processors is greater than or equal to
  $\frac{1+\varepsilon}{3}$, task $\tau_{N}$ can meet its deadline.
  Therefore, in our designed semi-partitioned schedule, each processor
  $m$ has a task set ${\bf T}_m$ that consists of three tasks:
  $\tau_m$ and $\tau_{m+M}$ from ${\bf T}^{ad}$ and a subtask
  $\tau_{N,m}$ of task $\tau_{N}$ with execution time $C_{N}/M$. We
  assign the second priority to task $\tau_{m+M}$ and the lowest
  priority to task $\tau_m$ on processor $m$.

  We utilize the worst-case response time analysis by Bini et
  al. \cite{DBLP:journals/tc/BiniNRB09}.  They showed that if
  $1-\sum_{\tau_i \in hp(\tau_k,m)} U_i\leq 1$, then the
  worst-case response time of a task $\tau_k$ in a task set ${\bf
    T}_m$ under fixed-priority scheduling on a processor is at most
  \begin{equation}
    \label{eq:WCRT-bini-quadratic}
    \frac{C_k+ \sum_{\tau_i \in hp(\tau_k,m)} C_i - \sum_{\tau_i \in hp(\tau_k,m) } U_i C_i}{1-\sum_{\tau_i \in hp(\tau_k,m)} U_i},
  \end{equation}
  where $hp(\tau_k, m)$ is the set of the tasks in ${\bf T}_m$ that
have a higher priorities than task $\tau_k$.
  Note that the precondition \mbox{$1-\sum_{\tau_i \in hp(\tau_k,m)} U_i\leq
  1$} for the test in Eq.~\eqref{eq:WCRT-bini-quadratic} to be
  applicable always holds
  at any arbitrarily speed since we assign $\tau_m$ as the
  lower-priority task on processor $m$ and $U_{m+M} \rightarrow 0$, and $U_{N,m} =
  C_{N,m}/T_N \rightarrow 0$.

  By Eq.~\eqref{eq:WCRT-bini-quadratic}, if the speed of processor $m$
  is greater than or equal to $\frac{C_{N}}{M} + C_{m+M} =
  \frac{1+\varepsilon}{3M}+\frac{1}{3}$, task $\tau_{m+M}$ can still meet
  its deadline in this schedule. By Eq.~\eqref{eq:WCRT-bini-quadratic}, 
  task $\tau_m$ can meet its deadline at speed $s$ in this schedule if
  \begin{align}
1 \geq &\frac{C_m/s+\sum_{\tau_i \in hp(\tau_k,m)} C_i/s - \sum_{\tau_i \in hp(\tau_k,m) } \frac{U_i}{s} \frac{C_i}{s}}{1-\sum_{\tau_i \in hp(\tau_k,m)} U_i/s}
=  \frac{\varepsilon}{3s}   + \frac{1}{3s} + \frac{1+\varepsilon}{3sM}
  \end{align}
  Therefore, as long as $s \geq \frac{1+\varepsilon}{3} +
  \frac{1+\varepsilon}{3M}$, task $\tau_m$ meets its deadline under our
  designed schedule.
\end{proof}



\end{document}